\documentclass[runningheads]{llncs}
\usepackage[T1]{fontenc}
\usepackage{graphicx} 
\usepackage{amsmath,amssymb,amsfonts}
\usepackage{algorithm}
\usepackage{algpseudocode}
\usepackage{tabularx}
\usepackage{xcolor}
\usepackage{booktabs}
\begin{document}
\title{Anamorphic Encryption with CCA Security: A Standard Model Construction}
%
%
\author{Shujun Wang\inst{1} \and
Jianting Ning\inst{2} \and 
Qinyi Li\inst{1} \and 
Leo Yu Zhang\inst{1}}
\authorrunning{S. Wang et al.} 

\institute{Griffith University, Australia \\
\email{\{shujunwang677@outlook.com, qinyi.li@griffith.edu.cn, leo.zhang@griffith.edu.au\}} \and
Wuhan University, China \\
\email{jtning88@gmail.com}}

\maketitle
                                                                                                                                                                                                                                                                                                                                                                                                                                            \begin{abstract}
Anamorphic encryption serves as a vital tool for covert communication, maintaining secrecy even during post-compromise scenarios. Particularly in the receiver-anamorphic setting, a user can shield hidden messages even when coerced into surrendering their secret keys. However, a major bottleneck in existing research is the reliance on CPA-security, leaving the construction of a generic, CCA-secure anamorphic scheme in the standard model as a persistent open challenge. To bridge this gap, we formalize the Anamorphic Key Encapsulation Mechanism (AKEM), encompassing both Public-Key (PKAKEM) and Symmetric-Key (SKAKEM) variants. We propose generic constructions for these primitives, which can be instantiated using any KEM that facilitates randomness recovery. Notably, our framework achieves strong IND-CCA (sIND-CCA) security for the covert channel. We provide a rigorous formal proof in the standard model, demonstrating resilience against a "dictator" who controls the decapsulation key. The security of our approach is anchored in the injective property of the base KEM, which ensures a unique mapping between ciphertexts and randomness. By integrating anamorphism into the KEM-DEM paradigm, our work significantly enhances the practical utility of covert channels within modern cryptographic infrastructures.

\end{abstract}
\section{Introduction}
The "Crypto Wars" represent a pivotal political and technological conflict with significant implications for the global digital landscape~\cite{rogaway2015moral}. At its core, the confrontation involves a vigorous debate between cryptographers and policymakers over the fundamental tension between "encryption restrictions" and "privacy protection." The role of cryptography in safeguarding privacy is fundamentally reliant on two key assumptions: Firstly, the sender-freedom assumption, which asserts that users can freely select their communication content without third-party interference; and secondly, the receiver-privacy assumption, which requires that the recipient’s private key remain strictly confidential. However, within the context of the Crypto Wars, dictator demands—such as those from powerful governments—are seeking to challenge these two foundational assumptions.

To address the disruptive challenges posed to foundational cryptographic assumptions in environments with heavy censorship, Persiano et al.~\cite{persiano2022anamorphic} introduced the paradigm of "Anamorphic Encryption (AME)." A core design principle of the paradigm acknowledges that deploying a new, dedicated encryption scheme for secret communication would readily arouse a dictator's suspicion and be subsequently blocked. Consequently, AME is ingeniously built upon existing, widely deployed Public-Key Encryption (PKE) systems. The objective is to enable communicating parties to embed a covert channel, imperceptible to the dictator, within a seemingly conventional encrypted channel for the transmission of covert information. Specifically, AME is categorized into two types: Sender-AME~\cite{wang2023sender}, which aims to counteract disruption of the sender-freedom assumption, and Receiver-AME~\cite{persiano2022anamorphic}, which focuses on scenarios where the receiver-privacy assumption is compromised. We focus on Receiver-AME, where parties supplement a standard $(pk, sk)$ pair with a pre-shared double key. Under coercive scrutiny, the receiver strategically yields only $sk$ as their purported sole secret; meanwhile, the undisclosed double key preserves the receiver's ability to recover covert messages. From the dictator's perspective, an anamorphic ciphertext that embeds a covert message must be computationally indistinguishable from a normal ciphertext generated with the public key $pk$.

In prior works, the double key was conceived as a symmetric shared secret ~\cite{persiano2022anamorphic} ~\cite{banfi2024anamorphic} ~\cite{kutylowski2023self}. Subsequently, the concept of Public-Key Anamorphic Encryption was introduced, wherein Persiano et al.~\cite{persiano2024public} proposed a variant where the double key is set to the empty string $\epsilon$ --- a design that intriguingly obviates the necessity of privately transmitting the double key to the receiver. More recently, Catalano et al.~\cite{catalano2024anamorphic} developed Fully Asymmetric Anamorphic Encryption, featuring a non-null double key structured as an asymmetric pair $(dk, tk)$.

However, regardless of whether the double key is symmetric or asymmetric, the vast majority of existing anamorphic schemes provide insufficient security guarantees for covert messages, predominantly settling for chosen-plaintext attack (CPA) security. Recent work by Banerjee et al.~\cite{banerjee2025simple} achieved Replayable Chosen-Ciphertext Attack (RCCA) security~\cite{canetti2003relaxing}~\cite{faonio2020improving} for covert messages in the asymmetric setting. However, their approach, based on the Fujisaki-Okamoto (FO) transform~\cite{fujisaki1999secure} within the Random Oracle Model (ROM)~\cite{canetti2004random}~\cite{don2019security}, suffers from limited generality and falls short of the ideal Chosen-Ciphertext Attack (CCA) security~\cite{canetti2004chosen}. Choi et al.~\cite{choi2025unified} pioneered the conceptual framework for CCA-security regarding covert messages, a significant contribution that addresses a theoretical void in anamorphic encryption. Despite its conceptual merits, the work provides neither a formalization of the construction nor a reduction-based security proof targeting that specific scheme. Furthermore, their methodology intrinsically relies on a monolithic PKE paradigm, which fundamentally diverges from modern cryptographic practice. In real-world deployments, public-key operations are almost exclusively instantiated via KEMs due to their efficiency and structural modularity in hybrid encryption. By treating the primitive as a traditional PKE, their design tightly couples the normal message $m$ and the covert message $m'$ to derive the encryption randomness $r$. This design introduces a critical vulnerability: any modification of $m$ by a dictator—even if $m'$ remains untouched—precludes the receiver from reconstructing $r$. This ``all-or-nothing'' dependency renders the covert channel excessively fragile. 
Given these cumulative limitations, the following question arises:

\begin{quote}
Is it possible to design a generic transformation for anamorphic encryption spanning both symmetric and asymmetric settings that achieves provable CCA security in the standard model without creating a fragile dependency between the normal and covert message spaces?
\end{quote}

\subsection{Our Contributions} 
In this work, we answer the question aforementioned in the affirmative sense. Our main contributions are summarized as follows:
\begin{itemize}
    \item We formalize the notions of Public-Key and Symmetric-Key Anamorphic Key Encapsulation Mechanism (PKAKEM and SKAKEM, respectively), which enable covert communication secure against a dictator holding the decapsulation key. By bridging the gap between anamorphism and the prevailing KEM-DEM paradigm, our work significantly improves the practical viability of anamorphic encryption for modern cryptographic ecosystems.
    \item We present concrete constructions for PKAKEM and SKAKEM, which can be instantiated from any KEM that supports randomness recovery. We formally prove the anamorphic security of our constructions, which ensures that normal and anamorphic ciphertexts are computationally indistinguishable from the dictator.
    \item We achieve strong Indistinguishability under Chosen-Ciphertext Attack (sIND-CCA) security for covert messages, rigorously proven in the standard model. By decoupling the normal message from the covert message, thereby overcoming the ``all-or-nothing'' dependency of prior designs and rendering covert retrieval resilient to alterations of the normal ciphertext.
\end{itemize}
\subsection{Technique Overview} 
The technical foundation of our paper lies in harmonizing the modern KEM-DEM paradigm with anamorphic cryptography. By leveraging KEM primitive~\cite{abe2005tag} ~\cite{nagao2005universally}~\cite{chen2020subvert}~\cite{dent2003designer}~\cite{saito2018tightly} as a building block, we develop PKAKEM and SKAKEM. The core technical challenge addressed in this work is the realization of strong Indistinguishability under Chosen-Ciphertext Attack (sIND-CCA) security within the standard model. In our context, "strong" security denotes that covert message confidentiality is preserved even against an adversary (e.g., a dictator) who possesses the legitimate decapsulation keys. Unlike prior attempts that suffer from an ``all-or-nothing'' dependency, our approach introduces a decoupling mechanism between the normal and covert message spaces. Such a design enables successful covert retrieval even when the normal message undergoes malicious alterations.

The security of an anamorphic system is anchored in the notion of anamorphic security. While standard KEM encapsulation utilizes fresh randomness to generate ciphertexts and keys, our approach uses this entropy as a covert channel. Specifically, we substitute the original random coins with a pseudorandom value deterministically derived from the covert message. The substitution is undetectable to a polynomial-time dictator, as the pseudorandom values effectively mimic the uniform randomness of an honest encapsulation, ensuring that the two distributions are computationally indistinguishable. In our concrete constructions, the pseudorandomness for PKAKEM is sourced from the pseudorandom properties of both the covert message's ciphertext and the accompanying Message Authentication Code (MAC), whereas SKAKEM relies on an Invertible Pseudorandom Function (IPF) to ensure the indistinguishability of its output.

However, achieving anamorphic security is merely the preliminary requirement. The core technical challenge is ensuring sIND-CCA security for covert messages against an active, decapsulation-key-holding dictator. A naive strategy might attempt to leverage the Existential Unforgeability against Chosen-Message Attacks (EUF-CMA) security of a MAC to detect ciphertext tampering; yet, this approach fails to prevent a "many-to-one" mapping problem. Specifically, an adversary could craft a related ciphertext that maps to the same internal pseudorandom coins as the challenge, causing the decapsulation oracle to leak critical information and leading to a complete collapse of CCA security. To mitigate this, Choi et al.~\cite{choi2025unified} proposed binding the normal message and the covert message within an authenticated encryption framework. However, such a rigid coupling introduces a critical fragility: any modification to the normal message by the dictator—even if the covert payload is untouched—precludes the reconstruction of the shared randomness, rendering the covert channel unusable.

Our solution is to deploy the anamorphic transformation within a KEM framework that supports randomness recovery. The intrinsic structure of such KEMs ensures a stronger, typically injective, binding between the ciphertext and the randomness used for its generation. This structural property fundamentally decouples the normal message from the covert message, breaking the rigid 'all-or-nothing' dependency that plagued prior designs. Beyond establishing a rigorous sIND-CCA security proof in the standard model, our construction achieves a seamless integration of anamorphism into mainstream cryptographic ecosystems. Such an alignment provides a novel and generalizable foundation for the eventual standardization and widespread deployment of anamorphic cryptography.

\section{Preliminaries}
\subsection{Notation}
Let $S_i$ denote the event that the adversary $\mathcal{A}$ succeeds in game $G_i$, and let $|\Pr[S_i]|$ be the probability of its occurrence. For a finite set $X$, the notation $x \xleftarrow{\$} X$ denotes that $x$ is sampled uniformly at random from $X$. Moreover, $y \leftarrow A^{\mathcal{O}_1, \mathcal{O}_2, \dots}$ indicates that $y$ is the output of the probabilistic algorithm $\mathcal{A}$ that gives the oracle access to $\mathcal{O}_1, \mathcal{O}_2$ and so on.
\subsection{Pseudorandom Function}
Let $F(k, x)$ be a Pseudorandom Function (PRF)~\cite{luby1988how}. We say that $F$ is a secure PRF if for every Probabilistic Polynomial-Time (PPT) distinguisher $\mathcal{D}$, the advantage of $\mathcal{D}$ in distinguishing between the output of $F$ (with a key $k$ chosen uniformly at random) and the output of a truly random function $f$ is negligible in the security parameter $\lambda$:
    \[
        \text{Adv}_{F,\mathcal{D}}^{\text{PRF}}(\lambda) = \left| \Pr[\mathcal{D}^{F(k, \cdot)}(1^\lambda) = 1] - \Pr[\mathcal{D}^{f(\cdot)}(1^\lambda) = 1] \right| \leq negl(\lambda).
    \]

\subsection{Message Authentication Code}
A Message Authentication Code (MAC)~\cite{bellare2000security} is a cryptographic primitive designed to ensure the integrity and authenticity of messages. Formally, a MAC scheme consists of the following three algorithms, defined over a key space $\mathcal{K}$, a message space $\mathcal{M}$ and a tag space $\mathcal{T}$:

\begin{itemize}
    \item{$k \leftarrow \mathsf{MAC.KGen}(1^\lambda$)}: The probabilistic key generation algorithm takes the security parameter $1^\lambda$ as input and returns a secret key $k \in \mathcal{K}$.
    \item{$\tau \leftarrow \mathsf{MAC.Tag}(k,m$)}: The probabilistic authentication algorithm takes a secret key $k$ and a message $m \in \mathcal{M}$ to produce an authentication tag $\tau \in \mathcal{T}$.
    \item{$b \leftarrow \mathsf{MAC.Verify}(k,m,\tau$)}: The deterministic verification algorithm takes a key $k$, a message $m$, and a tag $\tau$ as input, and outputs a bit $b \in \{0, 1\}$. Here, $b = 1$ indicates that the tag is valid, while $b = 0$ denotes invalidity.
\end{itemize}

A MAC scheme is considered secure if it provides \textit{Strong Existential Unforgeability against Chosen-Message Attacks (SUF-CMA)}~\cite{dodis2012message}. Formally, for any PPT adversary $\mathcal{A}$ with access to the authentication oracle $\mathcal{O}_{\mathsf{MAC.Tag}}(k, \cdot)$ and the verification oracle $\mathcal{O}_{\mathsf{MAC.Verify}}(k, \cdot, \cdot)$, the advantage of $\mathcal{A}$ in forging a valid message-tag pair is negligible with respect to the security parameter $\lambda$. Specifically, $\mathcal{A}$ succeeds if it outputs a pair $(m', \tau')$ such that:
\begin{itemize}
    \item $\mathsf{MAC.Verify}(k, m', \tau') = 1$;
    \item The specific pair $(m', \tau')$ was never returned by the authentication oracle $\mathcal{O}_{\mathsf{MAC.Tag}}(k, \cdot)$.
\end{itemize}

The SUF-CMA advantage of $\mathcal{A}$ is defined as:
$$
    \mathsf{Adv}_{\mathsf{MAC},\mathcal{A}}^{\mathsf{SUF\text{-}CMA}}(\lambda) = \Pr\left[ 
    \begin{array}{l}
        k \leftarrow \mathsf{MAC.KGen}(1^\lambda); \\
        (m', \tau') \leftarrow \mathcal{A}^{\mathcal{O}_{\mathsf{MAC.Tag}}(k, \cdot), \mathcal{O}_{\mathsf{MAC.Verify}}(k, \cdot, \cdot)}(1^\lambda) : \\
        \mathsf{MAC.Verify}(k, m', \tau') = 1 \land
        (m', \tau') \notin \mathcal{Q}
    \end{array}
    \right] \leq \mathsf{negl}(\lambda),
$$
where $\mathcal{Q}$ denotes the set of all message-tag pairs $(m, \tau)$ that were generated and returned by the authentication oracle $\mathcal{O}_{\mathsf{MAC.Tag}}$ during the game.

Additionally, we define the \textit{pseudorandomness} of a MAC scheme. For any PPT adversary $\mathcal{A}$, there exists a negligible function $negl(\lambda)$ such that for each security parameter $\lambda\in\mathbb{N}$, the advantage of $\mathcal{A}$ satisfies:

\begin{equation*}
\mathsf{Adv}_{\mathsf{MAC}, \mathcal{A}}^{\mathsf{pserand}}(\lambda) = 2 \cdot \left| \Pr \left[ 
\begin{array}{c}
k \leftarrow \mathsf{MAC.KGen}(1^\lambda); \\
m \leftarrow \mathcal{A}(1^\lambda); \\
\tau_0 \leftarrow \mathsf{MAC.Tag}(k, m); \\
\tau_1 \xleftarrow{\$} \{0, 1\}^{|\tau_0|}; \\
b \xleftarrow{\$} \{0, 1\}; \\
b' \leftarrow \mathcal{A}(\tau_b)
\end{array} : b = b' \right] - \frac{1}{2} \right| \leq \mathsf{negl}(\lambda).
\end{equation*}

\subsection{Public Key Encryption }
A public key encryption (PKE) scheme~\cite{canetti2003forward}  is formally defined by a tuple of three PPT algorithms $(\mathsf{PKE.KGen}, \mathsf{PKE.Enc}, \mathsf{PKE.Dec})$:

\begin{itemize}
\item{$(pk,sk)\leftarrow {\mathsf{PKE.KGen}}(1^\lambda$)}: The probabilistic key generation algorithm takes as input the security parameter $1^\lambda$ and outputs a public/secret key pair $(pk, sk)$.
\item{$C\leftarrow \mathsf{PKE.Enc}(pk,m;r$)}: The probabilistic encryption algorithm takes as input a public key $pk$, a message $m$ and a randomness $r$ to output a ciphertext $C$.
\item{$m\leftarrow \mathsf{PKE.Dec}(sk,C$)}: The deterministic decryption algorithm takes as input a secret key $sk$ and a ciphertext $C$. Then it outputs the message $m$.
\end{itemize}

\textbf{Correctness.}
    The PKE scheme satisfies correctness if for any key pair $(pk, sk)$ generated by $\mathsf{PKE.KGen}(1^\lambda$), any message $m$ in the message space $\mathcal{M}$, and any randomness $r$ used during encryption, it holds that:
    \[
    \mathsf{PKE.Dec}(sk, \mathsf{PKE.Enc}(pk, m; r)) = m.
\]

\textbf{Pseudorandomness.} A PKE scheme provides ciphertext pseudorandomness ~\cite{moller2004public} if, for any PPT adversary $\mathcal{A}$, there exists a negligible function $negl(\lambda)$ such that for each security parameter $\lambda\in\mathbb{N}$, the advantage of $\mathcal{A}$ satisfies:

\begin{equation*}
\mathsf{Adv}_{\mathsf{PKE}, \mathcal{A}}^{\mathsf{pserand}}(\lambda) = 2 \cdot \left| \Pr \left[ 
\begin{array}{c}
(pk, sk) \leftarrow \mathsf{PKE.KGen}(1^\lambda); \\
m \leftarrow \mathcal{A}(1^\lambda, pk); \\
C_0 \leftarrow \mathsf{PKE.Enc}(pk, m;r); \\
C_1 \xleftarrow{\$} \{0, 1\}^{|C_0|}; \\
b \xleftarrow{\$} \{0, 1\}; \\
b' \leftarrow \mathcal{A}(1^\lambda, pk, C_b)
\end{array} : b = b' 
\right] - \frac{1}{2} \right| \leq \mathsf{negl}(\lambda).
\end{equation*}

\textbf{IND-CCA Security.}
 A PKE scheme is IND-CCA secure if, for any PPT adversary $\mathcal{A}$, there exists a negligible function $negl(\lambda)$ such that for each security parameter $\lambda\in\mathbb{N}$, the advantage of $\mathcal{A}$ satisfies:

\[
\mathsf{Adv}^{\mathsf{IND\text{-}CCA}}_{\mathsf{PKE},\mathcal{A}}(\lambda) = 2 \cdot \left| \Pr\left[\mathsf{Expt}^{\mathsf{IND\text{-}CCA}}_{\mathsf{PKE},\mathcal{A}}(\lambda) = 1\right] - \frac{1}{2} \right| \leq \mathsf{negl}(\lambda),
\]
where the security experiment $\mathsf{Expt}^{\mathsf{IND\text{-}CCA}}_{\mathsf{PKE},\mathcal{A}}(\lambda)$ is formally defined as follows. It is mandated that $|m_0| = |m_1|$ and that $\mathcal{A}$ is strictly prohibited from querying the decryption oracle $\mathcal{O}_{\mathsf{Dec}}(sk, \cdot)$ on the challenge ciphertext $C_b$.

\begin{center}
\begin{minipage}{0.9\linewidth}
$\mathsf{Expt}^{\mathsf{IND\text{-}CCA}}_{\mathsf{PKE},\mathcal{A}}(\lambda)$
\hrule 
\vspace{0.5em}
\begin{algorithmic}[1]
    \State $(pk, sk) \leftarrow \mathsf{PKE.KGen}(1^\lambda)$
    \State $(m_0, m_1) \leftarrow \mathcal{A}^{\mathcal{O}_{\mathsf{Dec}}(sk, \cdot)}(pk)$ \Comment{Require $|m_0| = |m_1|$}
    \State $b \xleftarrow{\$} \{0, 1\}$
    \State $C_b\leftarrow \mathsf{PKE.Enc}(pk, m_b)$
    \State $b' \leftarrow \mathcal{A}^{\mathcal{O}_{\mathsf{Dec}}(sk, \cdot)}(C_b)$ \Comment{Cannot query $C_b$ to $\mathcal{O}_{\mathsf{Dec}}$}
    \State \Return $1$ if $b = b'$; otherwise, return $0$.
\end{algorithmic}
\end{minipage}
\end{center}

\subsection{Randomness-Recoverable Key Encapsulation Mechanisms}
A Randomness-Recoverable Key Encapsulation Mechanism (RR-KEM) is defined by a tuple of three polynomial-time algorithms over a session key space $\mathcal{K}$, a ciphertext space $\mathcal{C}$, and a randomness space $\mathcal{R}$:

\begin{itemize}
    \item{$(ek,dk) \leftarrow {\mathsf{KEM}_{RR}.\mathsf{KGen}}(1^\lambda$)}: The probabilistic key generation algorithm takes the security parameter $1^\lambda$ as input and outputs an encapsulation key $ek$ and a decapsulation key $dk$. 
    
    \item{$(K,C) \leftarrow {\mathsf{KEM}_{RR}.\mathsf{Encaps}}(ek,r_e$)}: The probabilistic encapsulation algorithm takes as input an  encapsulation key $ek$ and explicitly uses randomness $r_e \in \mathcal{R}$. It outputs a session key $K \in \mathcal{K}$ and a ciphertext $C\in \mathcal{C}$.
    \item{$(K,r_e)/\bot \leftarrow {\mathsf{KEM}_{RR}.\mathsf{Decaps}}(dk,C$)}: The deterministic decapsulation algorithm takes as input a decapsulation key $dk$ and a ciphertext $C$. It outputs the recovered session key $K$ and the underlying randomness $r_e$, or a rejection symbol $\bot$ if decapsulation fails.
\end{itemize}

\textbf{Correctness.}
 RR-KEM satisfies the correctness if, for any key pair $(ek, dk)$ generated by $\mathsf{KEM}_{RR}.\mathsf{KGen}(1^\lambda$)and any randomness $r_e \in \mathcal{R}$ used during encapsulation, it holds that:
    $$\mathrm{Pr}[\mathsf{KEM}_{RR}.\mathsf{Decaps}(dk,C) = (K,r_e) : (K,C) \leftarrow {\mathsf{KEM}_{RR}.\mathsf{Encaps}}(ek,r_e)] = 1.$$

\subsection{Invertible Pseudorandom Functions}
An Invertible Pseudorandom Function (IPF)~\cite{boneh2017constrained}  $F_{\mathsf{inv}}$ is defined over a key space $\mathcal{K}_I$, a domain $\mathcal{X}_I$ and a range $\mathcal{Y}_I$: $\mathcal{K}_I \times \mathcal{X}_I\rightarrow\mathcal{Y}_I.$ It comprises a probabilistic setup algorithm and two deterministic evaluation algorithms:
\begin{itemize}
    \item $k_I \leftarrow \mathsf{IPF.Setup}(1^\lambda)$: Takes the security parameter $1^\lambda$ as input and outputs a key $k_I \in \mathcal{K}_I$.
    \item $y \leftarrow F_{\mathsf{inv}}(k_I, x)$: Takes a key $k_I$ and an element $x \in \mathcal{X}_I$ as input, and outputs an element $y \in \mathcal{Y}_I$.
    \item $x / \bot \leftarrow F_{\mathsf{inv}}^{-1}(k_I, y)$: Takes a key $k_I$ and an element $y \in \mathcal{Y}_I$ as input, and outputs either an element $x \in \mathcal{X}_I$ or a rejection symbol $\bot$.
\end{itemize}

\textbf{Injectivity.} For every security parameter $\lambda$ and any key $k_I$ generated by $\mathsf{IPF.Setup}(1^\lambda)$, the function $F_{\mathsf{inv}}(k_I, \cdot)$ is strictly injective from $\mathcal{X}_I$ to $\mathcal{Y}_I$.

\textbf{Security.} An IPF is considered secure if, for any PPT adversary $\mathcal{A}$, its advantage in distinguishing $F_{\mathsf{inv}}$ from a truly random injective function is negligible:
\begin{equation*}
\begin{aligned}
\mathsf{Adv}^{\mathsf{IPF}}_{F_{\mathsf{inv}}, \mathcal{A}}(\lambda) 
&= \left| \Pr \left[ k_I \leftarrow \mathsf{IPF.Setup}(1^\lambda) : \mathcal{A}^{F_{\mathsf{inv}}(k_I, \cdot), F_{\mathsf{inv}}^{-1}(k_I, \cdot)}(1^\lambda) = 1 \right] \right. \\
&\quad - \left. \Pr \left[ R \xleftarrow{\$} \mathsf{InjFuns}[\mathcal{X}_I, \mathcal{Y}_I] : \mathcal{A}^{R(\cdot), R^{-1}(\cdot)}(1^\lambda) = 1 \right] \right| \leq \mathsf{negl}(\lambda),
\end{aligned}
\end{equation*}
where $\mathsf{InjFuns}[\mathcal{X}_I, \mathcal{Y}_I]$ denotes the set of all injective functions mapping from $\mathcal{X}_I$ to $\mathcal{Y}_I$. For any sampled random function $R$, its inverse $R^{-1}$ is defined such that it maps elements from $\mathcal{Y}_I$ back to $\mathcal{X}_I \cup \{\bot\}$. Specifically, $R^{-1}(y) = x$ if $R(x) = y$; otherwise, it returns $\bot$. Furthermore, when the domain and range coincide (i.e., $\mathcal{X}_I = \mathcal{Y}_I$), $\mathsf{InjFuns}[\mathcal{X}_I, \mathcal{Y}_I]$ represents the set of all permutations over $\mathcal{X}_I$.

\section{Anamorphic Key Encapsulation Mechanism}
\subsection{Public Key Anamorphic KEM}
\begin{definition}[Public Key Anamorphic KEM]
    A public key anamorphic KEM (PKAKEM) with a covert message space $\mathcal{M'}$ is defined by a tuple of three algorithms:
    \begin{itemize}
    \item{$(ek,dk,dk',tk') \leftarrow \mathsf{PKAKEM.aGen}(1^\lambda)$}: The probabilistic anamorphic key generation algorithm takes the security parameter $1^\lambda$ as input, and outputs a pair of standard encapsulation/decapsulation keys $(ek, dk)$, along with a public/secret key pair $(dk', tk')$ for the receiver of the covert channel.
    
    \item{$(K,act) \leftarrow \mathsf{PKAKEM.aEnc}(1^\lambda,ek,dk',amsg)$}: The probabilistic anamorphic encryption algorithm takes as input the security parameter $1^\lambda$, an encapsulation key $ek$, a covert receiver's public key $dk'$ and a covert message $amsg \in \mathcal{M'}$. It outputs a standard session key $K$ and an anamorphic ciphertext $act$.
    
    \item{$amsg/\bot \leftarrow \mathsf{PKAKEM.aDec}(dk, tk', act)$}: The deterministic anamorphic decryption algorithm takes as input a standard decapsulation key $dk$, a covert trapdoor key $tk'$, and an anamorphic ciphertext $act$. It outputs the recovered covert message $amsg$, or an error symbol $\bot$ if extraction fails.
\end{itemize}
\end{definition}

\textbf{Correctness.} 
PKAKEM satisfies correctness if, for any covert message $amsg\in\mathcal{M'}$, the probability of decapsulation failure is negligible:
\[
\mathrm{Pr}
\left[
\mathrm{amsg} \neq \mathsf{PKAKEM.aDec}(dk, tk', act) :
\begin{array}{l}
    (ek,dk,dk',tk') \leftarrow \mathsf{PKAKEM.aGen}(1^\lambda); \\
    (K,act) \leftarrow \mathsf{PKAKEM.aEnc}(1^\lambda,ek,dk',amsg)
\end{array}
\right]
\]
is negligible.

\textbf{Anamorphic security.} $\mathsf{PKAKEM}$ achieves anamorphic security (i.e., indistinguishability between normal and anamorphic ciphertexts) if, for any PPT adversary $\mathcal{A}$, the following advantage is negligible in $\lambda$:
\[
\mathsf{Adv}^{\mathsf{Ana\text{-}Security}}_{\mathsf{PKAKEM}, \mathcal{A}}(\lambda) = \left| \Pr[\mathsf{RealG}_{\mathsf{PKAKEM}}(\lambda, \mathcal{A}) = 1] - \Pr[\mathsf{AnamorphicG}_{\mathsf{PKAKEM}}(\lambda, \mathcal{A}) = 1] \right| \leq \mathsf{negl}(\lambda)
,\]
where the security games $\mathsf{RealG}_{\mathsf{PKAKEM}}$ and $\mathsf{AnamorphicG}_{\mathsf{PKAKEM}}$ are defined as follows:

\vspace{0.4cm}
\noindent
{\renewcommand{\arraystretch}{1.6}
\begin{tabular}{|l|}
    \hline
    $\mathsf{RealG}_{\mathsf{PKAKEM}}(\lambda, \mathcal{A})$ \\
    \hline
    $(ek, dk) \xleftarrow{\$} \mathsf{KEM}_{RR}.\mathsf{KGen}(1^\lambda)$ \\
    \textbf{return } $\mathcal{A}^{\mathcal{O}_e(ek, \cdot)}(ek, dk) \text{ where } \mathcal{O}_e(ek, amsg) \text{ computes } r_e \xleftarrow{\$} \mathcal{R} \text{ and returns } \mathsf{KEM}_{RR}.\mathsf{Encaps}(ek, r_e)$ \\
    \hline
\end{tabular}
}

\vspace{0.6cm}
\noindent
{\renewcommand{\arraystretch}{1.6}
\begin{tabular}{|l|}
    \hline
    $\mathsf{AnamorphicG}_{\mathsf{PKAKEM}}(\lambda, \mathcal{A})$ \\
    \hline
    $(ek, dk, dk', tk') \xleftarrow{\$} \mathsf{PKAKEM.aGen}(1^\lambda)$ \\
    \textbf{return } $\mathcal{A}^{\mathcal{O}_a(ek, dk', \cdot)}(ek, dk) \text{ where } \mathcal{O}_a(ek, dk', amsg) \text{ returns } \mathsf{PKAKEM.aEnc}(1^\lambda, ek, dk', amsg)$ \\
    \hline
\end{tabular}
}
\vspace{0.2cm}

\textbf{sIND-CCA security.}
PKAKEM achieves sIND-CCA security if, for any PPT adversary $\mathcal{A}$, there exists a negligible function $negl(\lambda)$ such that the advantage of $\mathcal{A}$ satisfies:
\[
\mathsf{Adv}^{\mathsf{sIND\text{-}CCA}}_{\mathsf{PKAKEM},\mathcal{A}}(\lambda) = 2 \cdot \left| \Pr[\mathsf{Expt}^{\mathsf{sIND\text{-}CCA}}_{\mathsf{PKAKEM},\mathcal{A}}(\lambda) = 1] - 1/2 \right| \leq \mathsf{negl}(\lambda),
\]
where the security experiment $\mathsf{Expt}^{\mathsf{sIND\text{-}CCA}}_{\mathsf{PKAK},\mathcal{A}}(\lambda)$ is formally defined as follows:

\begin{center}
\begin{minipage}{0.9\linewidth}
$\mathsf{Expt}^{\mathsf{sIND\text{-}CCA}}_{\mathsf{PKAKEM},\mathcal{A}}(\lambda)$
\hrule 
\vspace{0.5em}
\begin{algorithmic}[1]
    \State $(ek, dk, dk', tk') \leftarrow \mathsf{PKAKEM.aGen}(1^\lambda)$
    \State $(amsg_0, amsg_1) \leftarrow \mathcal{A}^{\mathsf{PKAKEM.aDec}(dk, tk', \cdot)}(ek, dk, dk')$ \Comment{$\mathcal{A}$ knows $dk$}
    \State $\beta \xleftarrow{\$} \{0, 1\}$
    \State $act_\beta \leftarrow \mathsf{PKAKEM.aEnc}(1^\lambda, ek, dk', amsg_\beta)$
    \State $\beta' \leftarrow \mathcal{A}^{\mathsf{PKAKEM.aDec}(dk, tk', \cdot)}(act_\beta, ek, dk, dk')$
    \State \Return 1 if $\beta = \beta'$ and $act_\beta$ was never queried to $\mathsf{aDec}$; otherwise 0.
\end{algorithmic}
\end{minipage}
\end{center}

\subsubsection{Specific Construction.}    
Let $\mathsf{PKE}$ denote a public key encryption scheme that provides ciphertext pseudorandomness and IND-CCA security. Let MAC be a message authentication code scheme that is pseudorandom and EUF-CMA secure. Let $\mathsf{KEM}_{RR}$ be a randomness-recoverable KEM. 

The (Encode, Decode) pair constitutes an invertible transformation designed to preserve the statistical properties of its input. Specifically, for any input $c'$, the encoded output $r_e = \mathsf{Encode}(c')$ is computationally indistinguishable from a uniform random string in the randomness space $\mathcal{R}$. Furthermore, the transformation is perfectly reversible, such that $\mathsf{Decode}(\mathsf{Encode}(c')) = c'$ holds for any valid input $c'$.

\begin{center}
\begin{minipage}{0.9\linewidth}
$\mathsf{PKAKEM.aGen}(1^\lambda)$
\hrule 
\vspace{0.5em}
\begin{algorithmic}[1]
    \State $(dk', tk') \gets \mathsf{PKE}.\mathsf{Gen}(1^\lambda)$
    \State $(ek, dk) \leftarrow \mathsf{KEM}_{RR}.\mathsf{KGen}(1^\lambda)$
    \State \Return $(ek, dk, dk', tk')$ 
\end{algorithmic}
\end{minipage}
\end{center}

\begin{center}
\begin{minipage}{0.9\linewidth}
$\mathsf{PKAKEM.aEnc}(1^\lambda, ek, dk', amsg)$
\hrule
\vspace{0.5em}
\begin{algorithmic}[1]
    \State $mak \gets \mathsf{MAC.KGen}(1^\lambda)$
    \State $c \gets \mathsf{PKE}_.\mathsf{Enc}(dk', amsg \parallel mak)$
    \State $\tau \gets \mathsf{MAC.Tag}(mak, c)$
    \State $c' \gets (c, \tau)$
    \State $r_e \gets \mathsf{Encode}(c')$
    \State $(K, act) \gets \mathsf{KEM}_{RR}.\mathsf{Encaps}(ek; r_e)$
    \State \Return $(K, act)$ 
\end{algorithmic}
\end{minipage}
\end{center}

\begin{center}
\begin{minipage}{0.9\linewidth}
$\mathsf{PKAKEM.aDec}(dk, tk', act)$
\hrule
\vspace{0.5em}
\begin{algorithmic}[1]
    \State $(K, r_e) \gets \mathsf{KEM}_{RR}.\mathsf{Decaps}(dk, act)$
    \State $c' \gets \mathsf{Decode}(r_e)$
    \State $\mathsf{Parse} \ c' \ \mathsf{as} \ (c, \tau)$
    \State $amsg \parallel mak \gets \mathsf{PKE}.\mathsf{Dec}(tk', c)$
    \State \textbf{if} $\mathsf{MAC.Verify}(mak, c, \tau) = 0$ \textbf{then}
    \State \hspace{\algorithmicindent} \Return $\perp$
    \State \Return $amsg$
\end{algorithmic}
\end{minipage}
\end{center}

\subsection{Security Analysis of PKAKEM}
\begin{theorem}
If $\mathsf{PKE}$ provides ciphertext pseudorandomness and $\mathsf{MAC}$ provides pseudorandomness, then the $\mathsf{PKAKEM}$ scheme achieves anamorphic security. 
\end{theorem}
\begin{proof}
We proceed via a sequence of computationally indistinguishable games $G_0, G_1$, and $G_2$. Let $S_i$ denote the event that the adversary $\mathcal{A}$ succeeds (i.e., outputs $1$) in game $G_i$, and let $\Pr[S_i]$ be the probability of its occurrence. 

\textbf{Game $G_0$:} This game corresponds exactly to the $\mathsf{AnamorphicG}_{\mathsf{PKAKEM}}(\lambda, \mathcal{A})$ game. The challenger answers $\mathcal{A}$'s oracle queries using the true anamorphic encryption algorithm. Thus:
$$
\Pr[\mathsf{AnamorphicG}_{\mathsf{PKAKEM}}(\lambda, \mathcal{A}) = 1] = \Pr[S_0].
$$

\textbf{Game $G_1$:} This game is identical to $G_0$, except that the oracle replaces the valid PKE ciphertext $c$ with a uniformly random string. Specifically, upon receiving a query $amsg$, the oracle computes $mak \leftarrow \mathsf{MAC.KGen}(1^\lambda)$, samples $c \xleftarrow{\$} \{0,1\}^{l_c}$ (where $l_c$ is the ciphertext length), computes $\tau \leftarrow \mathsf{MAC.Tag}(mak, c)$, encodes $r_e \leftarrow \mathsf{Encode}(c \parallel \tau)$, and returns $(K, act) \leftarrow \mathsf{KEM}_{RR}.\mathsf{Encaps}(ek; r_e)$.
\begin{lemma}
  For any PPT adversary $\mathcal{A}$, there exists a PPT algorithm $\mathcal{B}_1$ such that:
  $$
  | \Pr[S_0] - \Pr[S_1] | = \mathsf{Adv}^{\mathsf{pserand}}_{\mathsf{PKE}, \mathcal{B}_1}(\lambda).
  $$
\end{lemma}
\begin{proof}
    We construct a reduction algorithm $\mathcal{B}_1$ that plays the PKE pseudorandomness game against a challenger $\mathcal{C}_{\mathsf{PKE}}$.
    
    \textbf{Setup Phase:} $\mathcal{C}_{\mathsf{PKE}}$ generates $(pk, sk) \leftarrow \mathsf{PKE.KGen}(1^\lambda)$ and sends $pk$ to $\mathcal{B}_1$. $\mathcal{B}_1$ implicitly sets the anamorphic public key $dk' := pk$. It generates the standard KEM keys $(ek, dk) \leftarrow \mathsf{KEM}_{RR}.\mathsf{KGen}(1^\lambda)$ and invokes $\mathcal{A}^{\mathcal{O}_a(\cdot)}(ek, dk, dk')$.
    
    \textbf{Query Phase:} When $\mathcal{A}$ submits a query $amsg$, $\mathcal{B}_1$ simulates the oracle $\mathcal{O}_a$:
    \begin{enumerate}
        \item $\mathcal{B}_1$ generates $mak \leftarrow \mathsf{MAC.KGen}(1^\lambda)$.
        \item $\mathcal{B}_1$ constructs the message $m = amsg \parallel mak$ and submits it to $\mathcal{C}_{\mathsf{PKE}}$. 
        \item $\mathcal{C}_{\mathsf{PKE}}$ internally flips a coin $b \xleftarrow{\$} \{0,1\}$. It computes $C_0 \leftarrow \mathsf{PKE.Enc}(pk, m)$ and $C_1 \xleftarrow{\$} \{0,1\}^{|C_0|}$, and returns the challenge $C_b$ to $\mathcal{B}_1$.
        \item $\mathcal{B}_1$ receives $C_b$ and computes $\tau^* \leftarrow \mathsf{MAC.Tag}(mak, C_b)$.
        \item $\mathcal{B}_1$ computes $r_e^* \leftarrow \mathsf{Encode}(C_b \parallel \tau^*)$ and returns $(K, C) \leftarrow \mathsf{KEM}_{RR}.\mathsf{Encaps}(ek; r_e^*)$ to $\mathcal{A}$.
    \end{enumerate}
    
    \textbf{Guess Phase:} Eventually, $\mathcal{A}$ outputs a guess $b'$. $\mathcal{B}_1$ outputs $b'$ as its guess for $b$.
    
    If $\mathcal{C}_{\mathsf{PKE}}$ chose $b = 0$, $C_b$ is a valid ciphertext, and $\mathcal{A}$'s view is perfectly identical to $G_0$. If $\mathcal{C}_{\mathsf{PKE}}$ chose $b = 1$, $C_b$ is a uniformly random string, making $\mathcal{A}$'s view perfectly identical to $G_1$. According to a standard cryptographic identity, the difference in $\mathcal{A}$'s output probabilities between these two views is exactly the advantage function $2 \cdot |\Pr[b=b'] - 1/2|$ of $\mathcal{B}_1$. Thus, the difference $| \Pr[S_0] - \Pr[S_1] |$ is exactly bounded by $\mathsf{Adv}^{\mathsf{pserand}}_{\mathsf{PKE}, \mathcal{B}_1}(\lambda)$.
\end{proof}
\textbf{Game $G_2$:} This game modifies $G_1$ by also replacing the MAC tag $\tau$ with a uniformly random string. Upon a query $amsg$, the oracle samples $c \xleftarrow{\$} \{0,1\}^{l_c}$ and $\tau \xleftarrow{\$} \{0,1\}^{l_\tau}$ (where $l_\tau$ is the tag length), encodes $r_e \leftarrow \mathsf{Encode}(c \parallel \tau)$, and returns $(K, act) \leftarrow \mathsf{KEM}_{RR}.\mathsf{Encaps}(ek; r_e)$.

\begin{lemma}
     For any PPT adversary $\mathcal{A}$, there exists a PPT algorithm $\mathcal{B}_2$ such that:
     $$
     | \Pr[S_1] - \Pr[S_2] | = \mathsf{Adv}^{\mathsf{pserand}}_{\mathsf{MAC}, \mathcal{B}_2}(\lambda).
     $$
\end{lemma}
\begin{proof}
    We construct $\mathcal{B}_2$ interacting with a MAC pseudorandomness challenger $\mathcal{C}_{\mathsf{MAC}}$, which holds a hidden key $k \leftarrow \mathsf{MAC.KGen}(1^\lambda)$.
    
    \textbf{Setup Phase:} $\mathcal{B}_2$ independently generates $(ek, dk) \leftarrow \mathsf{KEM}_{RR}.\mathsf{KGen}(1^\lambda)$ and $(dk', tk') \leftarrow \mathsf{PKE.KGen}(1^\lambda)$. It then invokes $\mathcal{A}^{\mathcal{O}_a(\cdot)}(ek, dk, dk')$.
    
    \textbf{Query Phase:} For each query $amsg$ from $\mathcal{A}$, $\mathcal{B}_2$ simulates the oracle:
    \begin{enumerate}
        \item $\mathcal{B}_2$ samples a random string $c^* \xleftarrow{\$} \{0,1\}^{l_c}$ and submits $m = c^*$ to $\mathcal{C}_{\mathsf{MAC}}$.
        \item $\mathcal{C}_{\mathsf{MAC}}$ internally flips a coin $b \xleftarrow{\$} \{0,1\}$. It computes $\tau_0 \leftarrow \mathsf{MAC.Tag}(k, m)$ and $\tau_1 \xleftarrow{\$} \{0,1\}^{|\tau_0|}$, and returns the challenge $\tau_b$ to $\mathcal{B}_2$.
        \item $\mathcal{B}_2$ receives $\tau_b$, computes $r_e^* \leftarrow \mathsf{Encode}(c^* \parallel \tau_b)$, and returns $(K, C) \leftarrow \mathsf{KEM}_{RR}.\mathsf{Encaps}(ek; r_e^*)$ to $\mathcal{A}$.
    \end{enumerate}
    
    \textbf{Guess Phase:} $\mathcal{A}$ outputs $b'$. $\mathcal{B}_2$ outputs $b'$ as its guess for $b$.
    
    \textbf{Analysis:} If $b = 0$, $\tau_b$ is a valid MAC tag evaluated on the random string $c^*$, which perfectly matches $G_1$. If $b = 1$, $\tau_b$ is drawn uniformly at random, which perfectly matches $G_2$. Therefore, the difference $| \Pr[S_1] - \Pr[S_2] |$ is precisely bounded by $\mathcal{B}_2$'s advantage $\mathsf{Adv}^{\mathsf{pserand}}_{\mathsf{MAC}, \mathcal{B}_2}(\lambda)$.
\end{proof}

\textbf{Conclusion:} 
In $G_2$, the input to the encoding function is a bit-string $(c \parallel \tau)$ where both components are sampled uniformly at random ($c \xleftarrow{\$} \{0,1\}^{l_c}$ and $\tau \xleftarrow{\$} \{0,1\}^{l_\tau}$). By the implicit definition of Randomness-Recoverable KEMs, the mapping $\mathsf{Encode}$ deterministically transforms a uniformly distributed bit-string into uniformly distributed randomness $r_e \xleftarrow{\$} \mathcal{R}$. 

Consequently, the simulated oracle in $G_2$ is mathematically identical to an oracle that directly samples $r_e \xleftarrow{\$} \mathcal{R}$ and returns $\mathsf{KEM}_{RR}.\mathsf{Encaps}(ek; r_e)$, which perfectly corresponds to the $\mathsf{RealG}_{\mathsf{PKAKEM}}(\lambda, \mathcal{A})$ game. 
Therefore:
$$
\Pr[S_2] = \Pr[\mathsf{RealG}_{\mathsf{PKAKEM}}(\lambda, \mathcal{A}) = 1].
$$

By summing the probability bounds across the sequence of games using the triangle inequality, the adversary's total advantage is:
$$
\begin{aligned}
\mathsf{Adv}^{\mathsf{Ana\text{-}Security}}_{\mathsf{PKAKEM},\mathcal{A}}(\lambda) &= | \Pr[S_0] - \Pr[S_2] | \\
&\leq | \Pr[S_0] - \Pr[S_1] | + | \Pr[S_1] - \Pr[S_2] | \\
&= \mathsf{Adv}^{\mathsf{pserand}}_{\mathsf{PKE}, \mathcal{B}_1}(\lambda) + \mathsf{Adv}^{\mathsf{pserand}}_{\mathsf{MAC}, \mathcal{B}_2}(\lambda).
\end{aligned}
$$

Since both the $\mathsf{PKE}$ and $\mathsf{MAC}$ schemes provide pseudorandomness, their respective advantage functions are negligible in $\lambda$. Therefore, the total distinguishing advantage is negligible, which completes the proof.

\end{proof}

\begin{theorem}
Assume $\mathsf{PKE}$ is an IND-CCA secure public-key encryption scheme, and $\mathsf{MAC}$ is a SUF-CMA message authentication code. Then the $\mathsf{PKAKEM}$ scheme achieves sIND-CCA security in the standard model.
\end{theorem}
\begin{proof}
    We proceed via a sequence of games. Let $S_i$ denote the event that the adversary $\mathcal{A}$ successfully guesses the challenge bit (i.e., $\beta = \beta'$) in Game $G_i$. The advantage of $\mathcal{A}$ in Game $G_i$ is defined as $2 \cdot |\Pr[S_i] - 1/2|$.

    \textbf{Game $G_0$}: This is the original $\mathsf{Expt}^{\mathsf{sIND\text{-}CCA}}_{\mathsf{PKAKEM},\mathcal{A}}(\lambda)$ game. By definition, we have:
\[
\mathsf{Adv}^{\mathsf{sIND\text{-}CCA}}_{\mathsf{PKAKEM},\mathcal{A}}(\lambda) = 2 \cdot \left| \Pr[S_0] - 1/2 \right|.
\]

\textbf{Game $G_1$}: This game is identical to $G_0$, except we modify how the challenger answers decryption queries in the second phase (after the challenge ciphertext $act^\star$ is issued). When $\mathcal{A}$ submits a query $act_i \neq act^\star$, the challenger decodes it to $(c_i, \tau_i)$. If $c_i = c^\star$, the challenger immediately returns $\bot$ without performing further decryption.

\begin{lemma} \label{lem:game0_to_game1}
Let $S_0$ and $S_1$ be the events that the adversary $\mathcal{A}$ successfully guesses the challenge bit in Game $G_0$ and Game $G_1$, respectively. We have:
$$ |\Pr[S_0] - \Pr[S_1]| \leq \mathsf{Adv}^{\mathsf{SUF\text{-}CMA}}_{\mathsf{MAC}, \mathcal{B}}(\lambda). $$
\end{lemma}

\begin{proof}
Let $E$ be the event in Game $G_1$ where the adversary $\mathcal{A}$ submits a decryption query $act_i$ such that $act_i \neq act^\star$, but upon decoding, it yields $(c_i, \tau_i)$ with $c_i = c^\star$, and the MAC verification algorithm succeeds (i.e., $\mathsf{MAC.Verify}(mak^\star, c^\star, \tau_i) = 1$). 

Observe that Game $G_0$ and Game $G_1$ proceed perfectly identically unless event $E$ occurs. In Game $G_0$, such a query would be processed and potentially decrypted, whereas in Game $G_1$, it is immediately rejected with $\bot$. Conditional on $E$ not occurring ($\neg E$), the adversary's view in both games is identical, meaning the probability of $\mathcal{A}$ winning is identical:
\[
\Pr[S_0 \mid \neg E] = \Pr[S_1 \mid \neg E].
\]

Using the Law of Total Probability, we can expand $\Pr[S_0]$ and $\Pr[S_1]$:
\begin{align*}
\Pr[S_0] &= \Pr[S_0 \mid E]\Pr[E] + \Pr[S_0 \mid \neg E]\Pr[\neg E] \\
\Pr[S_1] &= \Pr[S_1 \mid E]\Pr[E] + \Pr[S_1 \mid \neg E]\Pr[\neg E]
\end{align*}

Subtracting the two equations, the terms conditioned on $\neg E$ perfectly cancel out:
\[
\Pr[S_0] - \Pr[S_1] = (\Pr[S_0 \mid E] - \Pr[S_1 \mid E]) \cdot \Pr[E]
\]

Taking the absolute value on both sides yields:
\[
|\Pr[S_0] - \Pr[S_1]| = |\Pr[S_0 \mid E] - \Pr[S_1 \mid E]| \cdot \Pr[E].
\]

Since probabilities are strictly bounded by the interval $[0, 1]$, the absolute difference $|\Pr[S_0 \mid E] - \Pr[S_1 \mid E]|$ can be at most $1$. Therefore, we derive the strict bound:
\[
|\Pr[S_0] - \Pr[S_1]| \leq \Pr[E].
\]

We now bound $\Pr[E]$ by constructing an algorithm $\mathcal{B}$ that uses $\mathcal{A}$ to win the SUF-CMA game against the $\mathsf{MAC}$ scheme. 

$\mathcal{B}$ simulates the CCA game for $\mathcal{A}$ exactly as in $G_0$. $\mathcal{B}$ generates the key pairs $(ek, dk)$ and $(dk', tk')$ by itself, allowing it to honestly answer all of $\mathcal{A}$'s decryption queries. During the challenge phase, $\mathcal{B}$ generates a random MAC key $mak^\star$, encrypts $amsg_\beta \parallel mak^\star$ to produce $c^\star$, computes $\tau^\star \gets \mathsf{MAC.Tag}(mak^\star, c^\star)$, and constructs the challenge ciphertext $act^\star$. 

If event $E$ occurs during the subsequent decryption query phase, $\mathcal{A}$ has submitted a ciphertext $act_i \neq act^\star$ that decodes to $(c^\star, \tau_i)$ such that 
$$\mathsf{MAC.Verify}(mak^\star, c^\star, \tau_i) = 1. $$
We analyze this submission:
\begin{itemize}
    \item By the definition of the scheme, the mappings provided by $\mathsf{Encode}$ and $\mathsf{KEM}_{RR}$ are injective. Since $act_i \neq act^\star$ but $c_i = c^\star$, it is mathematically guaranteed that $\tau_i \neq \tau^\star$.
    \item The only valid MAC tag for the message $c^\star$ that $\mathcal{A}$ has ever witnessed under the key $mak^\star$ is $\tau^\star$. 
\end{itemize}

Since $\tau_i \neq \tau^\star$ and $\mathsf{MAC.Verify}(mak^\star, c^\star, \tau_i) = 1$, the pair $(c^\star, \tau_i)$ constitutes a fresh, valid message-tag pair that was never outputted by the MAC generation algorithm. When this occurs, $\mathcal{B}$ immediately halts the simulation and outputs $(c^\star, \tau_i)$ as its forgery. 

Because $\mathcal{B}$'s simulation of $G_0$ is perfect up until the moment $E$ occurs, the probability that $\mathcal{A}$ triggers event $E$ is exactly the probability that $\mathcal{B}$ produces a valid SUF-CMA forgery. Thus:
\[
\Pr[E] \leq \mathsf{Adv}^{\mathsf{SUF\text{-}CMA}}_{\mathsf{MAC}, \mathcal{B}}(\lambda).
\]

We conclude the proof of the lemma:
\[
|\Pr[S_0] - \Pr[S_1]| \leq \mathsf{Adv}^{\mathsf{SUF\text{-}CMA}}_{\mathsf{MAC}, \mathcal{B}}(\lambda).
\]
\end{proof}

\textbf{Game $G_2$}: This game alters the challenge encryption phase. Instead of encrypting $amsg_\beta \parallel mak^\star$, the challenger constructs a dummy message $\theta = 0^{|amsg_\beta \parallel mak^\star|}$ (a string of zeros of equal length) and encrypts $\theta$ to generate $c^\star$. The rest of the encapsulation proceeds normally.

\begin{lemma} \label{lem:game1_to_game2}
Let $S_1$ and $S_2$ be the events that the adversary $\mathcal{A}$ successfully guesses the challenge bit in Game $G_1$ and Game $G_2$, respectively. We have:
\[
|\Pr[S_1] - \Pr[S_2]| \leq \mathsf{Adv}^{\mathsf{IND\text{-}CCA}}_{\mathsf{PKE},\mathcal{B}}(\lambda).
\]
\end{lemma}

\begin{proof}
Suppose there exists a polynomial-time adversary $\mathcal{A}$ that can distinguish Game $G_1$ from Game $G_2$. We construct a reduction algorithm $\mathcal{B}$ that uses $\mathcal{A}$ as a subroutine to break the standard IND-CCA security of the underlying $\mathsf{PKE}$ scheme.

\textbf{Setup Phase:} \\
The IND-CCA challenger for $\mathsf{PKE}$ generates a key pair $(pk,sk)$ and sends the public key $pk$ to $\mathcal{B}$. 
$\mathcal{B}$ embeds this public key into the anamorphic scheme by setting $dk' = pk$. Note that $\mathcal{B}$ does not know the corresponding secret key $tk'$ (which equals $sk$). $\mathcal{B}$ then natively generates the randomness-recoverable KEM key pair $(ek, dk) \gets \mathsf{KEM}_{RR}.\mathsf{KGen}(1^\lambda)$. $\mathcal{B}$ sends the public parameters $(ek, dk, dk')$ to $\mathcal{A}$.

\textbf{Pre-Challenge Decryption Queries:} \\
When $\mathcal{A}$ submits a decryption query for a ciphertext $act_i$, $\mathcal{B}$ must simulate the $\mathcal{O}_{\mathsf{aDec}}$ oracle without knowing $tk'$. $\mathcal{B}$ proceeds as follows:
\begin{enumerate}
    \item $\mathcal{B}$ uses its knowledge of $dk$ to decapsulate $act_i$: 
    $$(K_i, r_{e,i}) \gets \mathsf{KEM}_{RR}.\mathsf{Decaps}(dk, act_i).$$
    \item $\mathcal{B}$ decodes the randomness: $(c_i, \tau_i) \gets \mathsf{Decode}(r_{e,i})$.
    \item Since $\mathcal{B}$ lacks $tk'$, it forwards $c_i$ to its own $\mathsf{PKE}$ IND-CCA decryption oracle. The oracle returns the underlying plaintext $amsg_i \parallel mak_i$.
    \item $\mathcal{B}$ verifies the MAC: if $\mathsf{MAC.Verify}(mak_i, c_i, \tau_i) = 1$, it returns $amsg_i$ to $\mathcal{A}$; otherwise, it returns $\bot$.
\end{enumerate}
This simulation is perfectly indistinguishable from the real game.

\textbf{Challenge Phase:} \\
$\mathcal{A}$ outputs two equal-length messages $amsg_0$ and $amsg_1$. $\mathcal{B}$ generates a fresh MAC key $mak^\star \gets \mathsf{MAC.KGen}(1^\lambda)$ and picks a random bit $\beta \xleftarrow{\$} \{0,1\}$. 
$\mathcal{B}$ constructs two challenge plaintexts for its $\mathsf{PKE}$ challenger:
\begin{align*}
M_0 &= amsg_\beta \parallel mak^\star \\
M_1 &= \theta \quad \text{(where $\theta = 0^{|amsg_\beta \parallel mak^\star|}$)}
\end{align*}
$\mathcal{B}$ submits $(M_0, M_1)$ to its IND-CCA challenger. The challenger picks a random bit $b \xleftarrow{\$} \{0,1\}$, encrypts $M_b$ under $dk'$, and returns the challenge ciphertext $c^\star$. 
$\mathcal{B}$ computes $\tau^\star \gets \mathsf{MAC.Tag}(mak^\star, c^\star)$, sets $c^{\star\star} = (c^\star, \tau^\star)$, encodes $r_e^\star \gets \mathsf{Encode}(c^{\star\star})$, and encapsulates it: $(act^\star, K^\star) \gets \mathsf{KEM}_{RR}.\mathsf{Encaps}(ek; r_e^\star)$. $\mathcal{B}$ sends $act^\star$ to $\mathcal{A}$.

\textbf{Post-Challenge Decryption Queries:} \\
$\mathcal{A}$ may continue to make queries $act_i \neq act^\star$. $\mathcal{B}$ processes them exactly as in the pre-challenge phase, with one crucial exception strictly following the rule introduced in Game $G_1$: if decoding $act_i$ yields $c_i = c^\star$, $\mathcal{B}$ immediately returns $\bot$. 
This step is vital because $\mathcal{B}$ is prohibited from querying $c^\star$ to its own $\mathsf{PKE}$ decryption oracle. By leveraging the rule from $G_1$, $\mathcal{B}$ legally avoids illegal oracle queries while maintaining a perfect simulation.

\textbf{Guess:} \\
Eventually, $\mathcal{A}$ outputs a guess $\beta'$. $\mathcal{B}$ uses this to determine its own guess $b'$ for the PKE challenger:
\begin{itemize}
    \item If $\beta' = \beta$, $\mathcal{B}$ outputs $b' = 0$.
    \item If $\beta' \neq \beta$, $\mathcal{B}$ outputs $b' = 1$.
\end{itemize}

We now formally map the probabilities. By definition, the IND-CCA advantage of $\mathcal{B}$ against $\mathsf{PKE}$ is:
\[
\mathsf{Adv}^{\mathsf{IND\text{-}CCA}}_{\mathsf{PKE},\mathcal{B}}(\lambda) = \left| \Pr[b'=0 \mid b=0] - \Pr[b'=0 \mid b=1] \right|.
\]

When the PKE challenger's bit $b = 0$, the challenge ciphertext $c^\star$ is an encryption of $M_0 = amsg_\beta \parallel mak^\star$. In this scenario, $\mathcal{B}$ provides $\mathcal{A}$ with a view that is perfectly distributed identically to Game $G_1$. Thus, the probability that $\mathcal{A}$ guesses correctly ($\beta' = \beta$) is exactly the probability of event $S_1$:
\[
\Pr[b'=0 \mid b=0] = \Pr[\beta' = \beta \mid b=0] = \Pr[S_1].
\]

When the PKE challenger's bit $b = 1$, $c^\star$ is an encryption of $M_1 = \theta$. Here, the view of $\mathcal{A}$ is perfectly distributed identically to Game $G_2$. Thus, the probability that $\mathcal{A}$ guesses correctly is exactly the probability of event $S_2$:
\[
\Pr[b'=0 \mid b=1] = \Pr[\beta' = \beta \mid b=1] = \Pr[S_2].
\]

Substituting these into the advantage equation directly yields:
\[
\mathsf{Adv}^{\mathsf{IND\text{-}CCA}}_{\mathsf{PKE},\mathcal{B}}(\lambda) = \left| \Pr[S_1] - \Pr[S_2] \right|.
\]
\end{proof}

\textbf{Conclusion}: In Game $G_2$, the challenge ciphertext $act^\star$ encrypts the dummy string $\theta$ and is completely independent of the bit $\beta$. The adversary gains zero information about $\beta$ from the challenge. Thus, $\Pr[S_2] = 1/2$. 

Combining the bounds from the sequence of games via the triangle inequality, we obtain:
\begin{align*}
\mathsf{Adv}^{\mathsf{sIND\text{-}CCA}}_{\mathsf{PKAKEM},\mathcal{A}}(\lambda) &= 2 \cdot \left| \Pr[S_0] - 1/2 \right| \\
&\leq 2 \cdot \left| \Pr[S_0] - \Pr[S_1] \right| + 2 \cdot \left| \Pr[S_1] - \Pr[S_2] \right| + 2 \cdot \left| \Pr[S_2] - 1/2 \right| \\
&\leq 2 \cdot \mathsf{Adv}^{\mathsf{SUF\text{-}CMA}}_{\mathsf{MAC}, \mathcal{B}}(\lambda) + 2 \cdot \mathsf{Adv}^{\mathsf{IND\text{-}CCA}}_{\mathsf{PKE},\mathcal{B}}(\lambda).
\end{align*}
Since both terms on the right side are negligible, the advantage of any PPT adversary $\mathcal{A}$ is negligible.
\end{proof}

\subsection{Symmetric Key Anamorphic KEM}
\begin{definition}[Symmetric Key Anamorphic KEM]
    A Symmetric Key Anamorphic KEM (SKAKEM) with covert message space $\mathcal{M'}$ is defined by a triple of polynomial-time algorithms:
    \begin{itemize}
    \item{$(ek,dk,DK) \leftarrow \mathsf{SKAKEM.aGen}(1^\lambda)$}: The probabilistic anamorphic key generation algorithm inputs a security parameter $1^\lambda$ and outputs a standard encapsulation key $ek$, a standard decapsulation key $dk$ along with an anamorphic double key $DK$.
    \item{$(K,act) \leftarrow \mathsf{SKAKEM.aEnc}(ek,DK,amsg,ctr)$}: The probabilistic anamorphic encryption algorithm takes as input the encapsulation key $ek$, the double key $DK$, a covert message $amsg \in \mathcal{M'}$, and a counter $ctr$. It outputs a session key $K$ and an anamorphic ciphertext $act$.
    \item{$amsg/\bot \leftarrow \mathsf{SKAKEM.aDec}(DK, dk, ctr, act)$}: The deterministic anamorphic decryption algorithm inputs a double key $DK$, a decapsulation key $dk$, a matching counter $ctr$ and an anamorphic ciphertext $act$. It outputs the covert message $amsg$ or an error symbol $\bot$.
\end{itemize}
\end{definition}

\textbf{Correctness.} SKAKEM satisfies correctness if for any anamorphic message $amsg\in\mathcal{M'}$, it holds that:
\[
\mathrm{Pr}
\left[
\mathrm{amsg} \neq \mathsf{SKAKEM.aDec}(DK, dk, ctr, act) :
\begin{array}{l}
    (ek,dk,DK) \leftarrow \mathsf{SKAKEM.aGen}(1^\lambda); \\
    (K,act) \leftarrow \mathsf{SKAKEM.aEnc}(ek,DK, amsg, ctr)
\end{array}
\right]
\]
is negligible.

\textbf{Anamorphic security.} SKAKEM achieves anamorphic security (i.e., indistinguishability between normal and anamorphic ciphertexts) if, for any PPT adversary $\mathcal{A}$, the following advantage is negligible in $\lambda$:
\[
\mathsf{Adv}^{\mathsf{Ana\text{-}Security}}_{\mathsf{SKAKEM}, \mathcal{A}}(\lambda) = \left| \Pr[\mathsf{RealG}_{\mathsf{SKAKEM}}(\lambda, \mathcal{A}) = 1] - \Pr[\mathsf{AnamorphicG}_{\mathsf{SKAKEM}}(\lambda, \mathcal{A}) = 1] \right| \leq \mathsf{negl}(\lambda),
\]
where the security games $\mathsf{RealG}_{\mathsf{SKAKEM}}$ and $\mathsf{AnamorphicG}_{\mathsf{SKAKEM}}$ are defined as follows:

\vspace{0.4cm}
\noindent
{\renewcommand{\arraystretch}{1.6}
\begin{tabular}{|l|}
    \hline
    $\mathsf{RealG}_{\mathsf{SKAKEM}}(\lambda, \mathcal{A})$ \\
    \hline
    $(ek, dk) \xleftarrow{\$} \mathsf{KEM}_{RR}.\mathsf{KGen}(1^\lambda)$ \\
    \textbf{return } $\mathcal{A}^{\mathcal{O}_e(ek, \cdot, \cdot)}(ek, dk)$ \\
    \quad $\text{where } \mathcal{O}_e(ek, amsg, ctr) \text{ computes } r_e \xleftarrow{\$} \mathcal{R} \text{ and returns } \mathsf{KEM}_{RR}.\mathsf{Encaps}(ek; r_e)$ \\
    \hline
\end{tabular}
}

\vspace{0.6cm}
\noindent
{\renewcommand{\arraystretch}{1.6}
\begin{tabular}{|l|}
    \hline
    $\mathsf{AnamorphicG}_{\mathsf{SKAKEM}}(\lambda, \mathcal{A})$ \\
    \hline
    $(ek, dk, DK) \xleftarrow{\$} \mathsf{SKAKEM.aGen}(1^\lambda)$ \\
    \textbf{return } $\mathcal{A}^{\mathcal{O}_a(ek, DK, \cdot, \cdot)}(ek, dk)$ \\
    \quad $\text{where } \mathcal{O}_a(ek, DK, amsg, ctr) \text{ returns } \mathsf{SKAKEM.aEnc}(ek, DK, amsg, ctr)$ \\
    \hline
\end{tabular}
}
\vspace{0.2cm}

\textbf{sIND-CCA security.}
$\mathsf{SKAKEM}$ achieves sIND-CCA security if, for any PPT adversary $\mathcal{A}$, there exists a negligible function $\mathsf{negl}(\lambda)$ such that the advantage of $\mathcal{A}$ satisfies:
\[
\mathsf{Adv}^{\mathsf{sIND\text{-}CCA}}_{\mathsf{SKAKEM}, \mathcal{A}}(\lambda) = 2 \cdot \left| \Pr[\mathsf{Expt}^{\mathsf{sIND\text{-}CCA}}_{\mathsf{SKAKEM}, \mathcal{A}}(\lambda) = 1] - 1/2 \right| \leq \mathsf{negl}(\lambda),
\]
where the security experiment $\mathsf{Expt}^{\mathsf{sIND\text{-}CCA}}_{\mathsf{SKAKEM}, \mathcal{A}}(\lambda)$ is formally defined as follows:

\vspace{0.4cm}
\noindent
\begin{tabular}{@{}l@{}}
$\mathsf{Expt}^{\mathsf{sIND\text{-}CCA}}_{\mathsf{SKAKEM}, \mathcal{A}}(\lambda)$ \\
\hline
\rule{0pt}{3ex}1: $(ek, dk, DK) \xleftarrow{\$} \mathsf{SKAKEM.aGen}(1^\lambda)$ \\
2: $(amsg_0, amsg_1, ctr^*) \leftarrow \mathcal{A}^{\mathsf{SKAKEM.aEnc}(ek, DK, \cdot, \cdot), \mathsf{SKAKEM.aDec}(DK, dk, \cdot, \cdot)}(ek, dk)$ \quad $\triangleright$ $\mathcal{A}$ knows $dk$ \\
3: $\beta \xleftarrow{\$} \{0,1\}$ \\
4: $(K^*, act^*) \leftarrow \mathsf{SKAKEM.aEnc}(ek, DK, amsg_\beta, ctr^*)$ \\
5: $\beta' \leftarrow \mathcal{A}^{\mathsf{SKAKEM.aEnc}(ek, DK, \cdot, \cdot), \mathsf{SKAKEM.aDec}(DK, dk, \cdot, \cdot)}(act^*, ek, dk)$ \\
6: \textbf{return} 1 if $\beta = \beta'$ and the following conditions hold: \\
\quad \quad (i) $ctr^*$ was never queried to $\mathsf{SKAKEM.aEnc}$ \\
\quad \quad (ii) $(ctr^*, act^*)$ was never queried to $\mathsf{SKAKEM.aDec}$ \\
\quad otherwise return 0.
\end{tabular}
\vspace{0.2cm}

\subsubsection{Specific Construction.}   
Let $F$ denote a secure pseudorandom function, and $F_{Inv}$ be a secure IPF. Let $ctr$ be a counter shared between the parties engaging in covert communication. Let MAC be a scheme that is both pseudorandom and SUF-CMA secure. Let $\text{KEM}_{RR}$ be a randomness-recoverable KEM.  
\vspace{0.3cm}
\begin{center}
\begin{minipage}{0.9\linewidth}
$\mathsf{SKAKEM.aGen}(1^\lambda)$
\hrule
\vspace{0.5em}
\begin{algorithmic}[1]
     \State $k \xleftarrow{\$} \mathcal{K}$
     \State $(ek, dk) \leftarrow \mathsf{KEM}_{RR}.\mathsf{KGen}(1^\lambda)$
     \State $mak \leftarrow \mathsf{MAC.KGen}(1^\lambda)$
     \State $DK := (k, mak)$
     \State \textbf{return} $(ek, dk, DK)$
\end{algorithmic}
\end{minipage}
\end{center}

\vspace{0.3cm}
\begin{center}
\begin{minipage}{0.9\linewidth}
$\mathsf{SKAKEM.aEnc}(ek, DK, amsg, ctr)$
\hrule
\vspace{0.5em}
\begin{algorithmic}[1]
     \State Parse $DK$ as $(k, mak)$
     \State $c := amsg \oplus F(k, ctr)$
     \State $\tau \leftarrow \mathsf{MAC.Tag}(mak, c)$
     \State $c' := (c, \tau)$
     \State $r_e \leftarrow F_{Inv}(k, c')$
     \State $(K, act) \leftarrow \mathsf{KEM}_{RR}.\mathsf{Encaps}(ek; r_e)$
     \State \textbf{return} $(K, act)$
\end{algorithmic}
\end{minipage}
\end{center}

\vspace{0.3cm}
\begin{center}
\begin{minipage}{0.9\linewidth}
$\mathsf{SKAKEM.aDec}(DK, dk, ctr, act)$
\hrule
\vspace{0.5em}
\begin{algorithmic}[1]
     \State Parse $DK$ as $(k, mak)$
     \State $(K, r_e) \leftarrow \mathsf{KEM}_{RR}.\mathsf{Decaps}(dk, act)$
     \State $c' := F^{-1}_{Inv}(k, r_e)$
     \State Parse $c'$ as $(c, \tau)$
     \State \textbf{if} $\mathsf{MAC.Verify}(mak, c, \tau) = 0$ \textbf{then}
     \State \hspace{\algorithmicindent} \textbf{return} $\bot$
     \State \textbf{return} $c \oplus F(k, ctr)$
\end{algorithmic}
\end{minipage}
\end{center}

\subsection{Security Analysis of SKAKEM}
\begin{theorem}
Assume $F$ is a secure pseudorandom function and $F_{Inv}$ is a secure inverse pseudorandom function. Assuming $F$ and $F_{Inv}$ employ standard domain separation, the $\mathsf{SKAKEM}$ scheme achieves anamorphic security. 
\end{theorem}
\begin{proof}
We proceed via a sequence of computationally indistinguishable games, $G_0$ and $G_1$. Let $S_i$ denote the event that a PPT adversary $\mathcal{A}$ outputs $1$ in game $G_i$.

\textbf{Game $G_0$:} This corresponds exactly to the $\mathsf{AnamorphicG}_{\mathsf{SKAKEM}}(\lambda, \mathcal{A})$ game. 
\[
\Pr[\mathsf{AnamorphicG}_{\mathsf{SKAKEM}}(\lambda, \mathcal{A}) = 1] = \Pr[S_0].
\]

\textbf{Game $G_1$:} In this game, we replace the evaluation of the IPF $F_{Inv}(k, \cdot)$ inside the encryption oracle with a truly random function $R(\cdot) : \{0,1\}^* \to \mathcal{R}$. 
Specifically, upon receiving a query $(amsg, ctr)$, the modified oracle computes $c := amsg \oplus F(k, ctr)$ and $\tau \leftarrow \mathsf{MAC.Tag}(mak, c)$. It sets $c' := (c, \tau)$, samples the randomness as $r_e \leftarrow R(c')$, and returns $\mathsf{KEM}_{RR}.\mathsf{Encaps}(ek; r_e)$.

\begin{lemma}
If $F_{Inv}$ is a secure PRF, there exists a PPT algorithm $\mathcal{B}$ such that:
\[
| \Pr[S_0] - \Pr[S_1] | = \mathsf{Adv}^{\mathsf{IPF}}_{F_{Inv}, \mathcal{B}}(\lambda).
\]
\end{lemma}
\begin{proof}
    We construct a reduction $\mathcal{B}$ interacting with a PRF challenger $\mathcal{C}_{IPF}$. Since $F$ and $F_{Inv}$ employ domain separation, they behave as independent PRFs. Thus, $\mathcal{B}$ independently samples the masking key $k_{enc} \xleftarrow{\$} \mathcal{K}$ and $mak \leftarrow \mathsf{MAC.KGen}(1^\lambda)$, while the challenger $\mathcal{C}_{IPF}$ holds the hidden key for $F_{Inv}$.
    
    When $\mathcal{A}$ submits $(amsg, ctr)$, $\mathcal{B}$ computes $c := amsg \oplus F(k_{enc}, ctr)$ and $\tau \leftarrow \mathsf{MAC.Tag}(mak, c)$. $\mathcal{B}$ sets $c' := (c, \tau)$ and submits $c'$ to its challenger $\mathcal{C}_{IPF}$, receiving a challenge response $r_e^*$. $\mathcal{B}$ then returns $\mathsf{KEM}_{RR}.\mathsf{Encaps}(ek; r_e^*)$ to $\mathcal{A}$.
    
    If $\mathcal{C}_{IPF}$ instantiates the real $F_{Inv}$, $\mathcal{A}$'s view is identical to $G_0$. If $\mathcal{C}_{IPF}$ uses a truly random function $R$, $\mathcal{A}$'s view perfectly matches $G_1$. Thus, $\mathcal{B}$'s advantage bounds the difference perfectly.
\end{proof}

\textbf{Analysis of Game $G_1$ to $\mathsf{RealG}$:}
In $G_1$, the randomness $r_e$ is generated by a true random function $R(c')$. If the sequence of queried inputs $c'$ contains no duplicates, then $R(c')$ yields perfectly independent and uniformly distributed values in $\mathcal{R}$, rendering $G_1$ identical to the $\mathsf{RealG}_{\mathsf{SKAKEM}}$ game (where $r_e \xleftarrow{\$} \mathcal{R}$ is explicitly sampled). 

We bound the probability of a collision in $c'$. A collision $c'_i = c'_j$ (for $i \neq j$) implies $c_i = c_j$. Recall that $c_i = amsg_i \oplus F(k, ctr_i)$. Since $\mathcal{A}$ is nonce-respecting, $ctr_i \neq ctr_j$. Due to the PRF security of $F(k, \cdot)$, the mask $F(k, ctr_i)$ is computationally unpredictable. Thus, $\mathcal{A}$ cannot adversarially choose $amsg_j$ to force $c_j = c_i$. The probability of an accidental collision is bounded by the PRF advantage plus the statistical collision bound $q_e^2 / 2^{|c|}$, where $q_e$ is the maximum number of queries. 

Therefore, $G_1$ is statistically close to the real game:
\[
| \Pr[S_1] - \Pr[\mathsf{RealG}_{\mathsf{SKAKEM}}(\lambda, \mathcal{A}) = 1] | \leq \mathsf{Adv}^{\mathsf{PRF}}_{F, \mathcal{A}}(\lambda) + \frac{q_e^2}{2^{|c|}}.
\]

By summing the probability bounds across the games via the triangle inequality, we have:
\[
\mathsf{Adv}^{\mathsf{Ana\text{-}Security}}_{\mathsf{SKAKEM}, \mathcal{A}}(\lambda) \leq \mathsf{Adv}^{\mathsf{IPF}}_{F_{Inv}, \mathcal{B}}(\lambda) + \mathsf{Adv}^{\mathsf{PRF}}_{F, \mathcal{A}}(\lambda) + \frac{q_e^2}{2^{|c|}}.
\]
Since $F_{Inv}$ and $F$ are secure PRFs, their advantages are negligible. For a sufficiently large ciphertext length $|c|$, $q_e^2 / 2^{|c|}$ is also negligible. Consequently, the total advantage is negligible, completing the proof.
\end{proof}

\begin{theorem}
Consider $F$ as a secure PRF and assume that the MAC scheme is SUF-CMA secure. Additionally, let  $F_{Inv}$ be a secure IPF and $\mathsf{KEM}_{RR}$ be a randomness-recoverable KEM. Then the $\mathsf{SKAKEM}$ scheme achieves sIND-CCA security.
\end{theorem}
\begin{proof}
   We proceed via a sequence of computationally indistinguishable games, $G_0, G_1$, and $G_2$. Let $S_i$ denote the event that a nonce-respecting PPT adversary $\mathcal{A}$ outputs $\beta' = \beta$ in game $G_i$.

\textbf{Game $G_0$:} This is the original $\mathsf{Expt}^{\mathsf{sIND\text{-}CCA}}_{\mathsf{SKAKEM}, \mathcal{A}}(\lambda)$ game as defined in the security model. 
\[
\mathsf{Adv}^{\mathsf{sIND\text{-}CCA}}_{\mathsf{SKAKEM}, \mathcal{A}}(\lambda) = 2 \cdot \left| \Pr[S_0] - 1/2 \right|.
\]

\textbf{Game $G_1$:} We modify the behavior of the decryption oracle. The challenger maintains a list $L_{enc}$ recording all ciphertexts generated by the encryption oracle $\mathsf{SKAKEM.aEnc}$. When $\mathcal{A}$ submits a query $(ctr, act)$ to the $\mathsf{SKAKEM.aDec}$ oracle, the challenger first checks if $act \in L_{enc}$. If $act \notin L_{enc}$, the oracle immediately returns $\bot$ without performing any decryption operations. If $act \in L_{enc}$, it returns the corresponding $amsg$ from its internal records.

\begin{lemma}
$| \Pr[S_0] - \Pr[S_1] | \leq \mathsf{Adv}^{\mathsf{SUF\text{-}CMA}}_{\mathsf{MAC}, \mathcal{B}_1}(\lambda)$.
\end{lemma}
\begin{proof}
    Let $Forge$ be the event that $\mathcal{A}$ queries $\mathsf{SKAKEM.aDec}$ with a valid ciphertext $act \notin L_{enc}$ that successfully passes the internal $\mathsf{MAC.Verify}$ check in $G_0$. Games $G_0$ and $G_1$ proceed identically until $Forge$ occurs. Thus, $| \Pr[S_0] - \Pr[S_1] | \leq \Pr[Forge]$. 
    
    We construct a reduction $\mathcal{B}_1$ to bound $\Pr[Forge]$. $\mathcal{B}_1$ interacts with a $\mathsf{SUF\text{-}CMA}$ challenger, gaining access to $\mathsf{MAC.Tag}$ and $\mathsf{MAC.Verify}$ oracles under a hidden key $mak^*$. $\mathcal{B}_1$ independently generates $k \xleftarrow{\$} \mathcal{K}$ and $(ek, dk) \leftarrow \mathsf{KEM}_{RR}.\mathsf{KGen}(1^\lambda)$, implicitly setting $DK = (k, mak^*, ek)$, and invokes $\mathcal{A}$.
    
    For queries $(amsg, ctr)$ to $\mathsf{SKAKEM.aEnc}$, $\mathcal{B}_1$ computes $c := amsg \oplus F(k, ctr)$, queries its $\mathsf{MAC.Tag}$ oracle to get $\tau$, computes $r_e \leftarrow F_{Inv}(k, (c, \tau))$, and encapsulates it to $act$, which is added to $L_{enc}$. 
    
    For queries $(ctr, act)$ to $\mathsf{SKAKEM.aDec}$, $\mathcal{B}_1$ decapsulates $act$ to retrieve $r_e$ and inverses it via $F_{Inv}^{-1}$ to get $(c, \tau)$. It then queries the $\mathsf{MAC.Verify}$ oracle with $(c, \tau)$. Since $\mathsf{KEM}_{RR}$ and $F_{Inv}$ mathematically establish an injective mapping from $act$ to $(c, \tau)$, any queried $act \notin L_{enc}$ implies a novel pair $(c, \tau)$ that was never generated by the $\mathsf{MAC.Tag}$ oracle. If $\mathsf{MAC.Verify}$ returns $1$, $\mathcal{B}_1$ has successfully forged a strong MAC and outputs $(c, \tau)$ as its winning forgery. Therefore, $\Pr[Forge]$ is exactly bounded by $\mathcal{B}_1$'s $\mathsf{SUF\text{-}CMA}$ advantage.
\end{proof}

\textbf{Game $G_2$:} In this game, we modify the computation of the challenge ciphertext. We replace the evaluation of the PRF $F(k, ctr^*)$ with a uniformly random string $R^* \xleftarrow{\$} \{0,1\}^{|amsg|}$. The challenge ciphertext is constructed using $c^* := amsg_\beta \oplus R^*$.

\begin{lemma}
$| \Pr[S_1] - \Pr[S_2] | \leq \mathsf{Adv}^{\mathsf{PRF}}_{F, \mathcal{B}_2}(\lambda)$.
\end{lemma}
\begin{proof}
    We construct $\mathcal{B}_2$ interacting with a PRF challenger. $\mathcal{B}_2$ holds a freshly generated $mak$ and $(ek, dk)$, but queries the PRF challenger for evaluations of $F(k, \cdot)$. 
    
    During $\mathsf{SKAKEM.aEnc}$ queries, $\mathcal{B}_2$ queries the PRF challenger with $ctr$ to obtain the mask, and honestly simulates the rest of the encryption. 
    Crucially, because we are in $G_1$, the $\mathsf{SKAKEM.aDec}$ oracle automatically rejects any $act \notin L_{enc}$. For $act \in L_{enc}$, $\mathcal{B}_2$ simply looks up the associated $amsg$ from its encryption history. Thus, $\mathcal{B}_2$ \textit{never needs to evaluate $F(k, ctr)$ to answer decryption queries}, completely avoiding any circular dependency with its challenger.
    
    For the challenge query $(amsg_0, amsg_1, ctr^*)$, $\mathcal{B}_2$ queries the PRF challenger with $ctr^*$ to get a string $Y^*$, and sets $c^* := amsg_\beta \oplus Y^*$. It computes the tag $\tau^*$ and encapsulates the challenge ciphertext $(K^*, act^*)$ normally.
    If the PRF challenger uses the real function $F$, this perfectly simulates $G_1$. If the challenger returns a truly random string, it perfectly simulates $G_2$. Therefore, the distinguishing advantage is bounded by $\mathsf{Adv}^{\mathsf{PRF}}_{F, \mathcal{B}_2}(\lambda)$.
\end{proof}

\textbf{Conclusion:} 
In $G_2$, the challenge ciphertext incorporates $c^* = amsg_\beta \oplus R^*$. Since $\mathcal{A}$ is required to be nonce-respecting, the specific challenge counter $ctr^*$ is never queried to the $\mathsf{SKAKEM.aEnc}$ oracle, meaning $R^*$ is a completely fresh, uniformly random string. Consequently, $R^*$ acts as a perfect one-time pad, information-theoretically hiding the bit $\beta$. The adversary's probability of guessing $\beta$ is exactly $1/2$, meaning $\Pr[S_2] = 1/2$.

Summing the bounds across the sequence of games yields:
\[
\mathsf{Adv}^{\mathsf{sIND\text{-}CCA}}_{\mathsf{SKAKEM}, \mathcal{A}}(\lambda) \leq 2 \cdot \left( \mathsf{Adv}^{\mathsf{SUF\text{-}CMA}}_{\mathsf{MAC}, \mathcal{B}_1}(\lambda) + \mathsf{Adv}^{\mathsf{PRF}}_{F, \mathcal{B}_2}(\lambda) \right) \leq \mathsf{negl}(\lambda).
\]
\end{proof}

\section{Discussion on RR-KEMs}
The anamorphic cryptographic framework proposed in this work inherently relies on the underlying Key Encapsulation Mechanism possessing the Randomness-Recoverable property. To demonstrate the broad applicability and practical impact of our design, we discuss various concrete instantiations of $\mathsf{KEM}_{RR}$. As summarized in Table~\ref{tab:instantiations}, we categorize these instantiations across two crucial dimensions: the underlying security model (Random Oracle Model vs. Standard Model) and the foundational hardness assumptions (Factoring/RSA, Discrete Logarithm, and Lattices/LPN).

\begin{table}[htbp]
    \centering
    \caption{Instantiations of RR-KEMs across different models and cryptographic assumptions.}
    \label{tab:instantiations}
    \renewcommand{\arraystretch}{1.3}
    \begin{tabular}{@{}lccc@{}}
        \toprule
        \textbf{Model} & \textbf{Factoring / RSA} & \textbf{Discrete Logarithm (DL)} & \textbf{Lattices / LPN} \\
        \midrule
        \textbf{ROM} & RSA-OAEP (PKCS\#1)\cite{rfc3447} & PSEC-KEM (ISO/IEC 18033-2)\cite{iso18033-2} & ML-KEM (FIPS 203) \cite{fips203}\\
        \textbf{Standard} & ATF-based schemes\cite{KMO10} & LTF-based schemes\cite{PW08} & Hybrid Encryption \cite{boyen2021secure} \\
        \bottomrule
    \end{tabular}
\end{table}

\subsection{Instantiations in the Random Oracle Model (ROM)}

In the ROM, IND-CCA secure KEMs are typically constructed by applying the FO transformation \cite{fujisaki1999secure} or its modern variants to a weakly secure PKE scheme. A core feature of modern FO transforms is \textit{derandomization} during encapsulation, which naturally yields the randomness-recoverability property.

\paragraph{RSA and Factoring-based Schemes.} 
RR-KEMs represent a foundational component in applied cryptography, their functionality often embedded within widely adopted standards without being explicitly designated as such. A canonical illustration of this principle is the key transport mechanism in legacy versions of the Transport Layer Security (TLS) ~\cite{rescorla2018transport}~\cite{blake2006transport}~\cite{eastlake2011transport} protocol. In versions such as TLS 1.2 ~\cite{dierks2008transport} and its predecessors, RSA-based key exchange cipher suites (e.g., \texttt{TLS\_RSA\_WITH\_AES\_256\_GCM\_SHA384}) were predominant. The protocol dictates that the client generates a random \textit{Pre-Master Secret}, encrypts it using the server's public RSA key, and transmits the resulting ciphertext. The server then decrypts this ciphertext to retrieve the original \textit{Pre-Master Secret}. This entire process is essentially an instance of a KEM, where the \textit{Pre-Master Secret} serves as the recoverable randomness.

The canonical cryptographic formalization of an RR-KEM is RSA-OAEP\cite{fujisaki2001rsa} \cite{kiltz2017instantiability}, standardized in PKCS\#1 \cite{rfc3447}. Its underlying encapsulation phase can be rigorously described by the following sequential steps:
\begin{enumerate}
    \item \textbf{Generation of Randomness:} The encapsulator generates a random byte string $r$, which constitutes the core randomness intended to be recovered by the recipient.
    \item \textbf{Key Derivation:} A standard Key Derivation Function~\cite{yao2005design}, such as HKDF ~\cite{krawczyk2010hmac}~\cite{krawczyk2010cryptographic}, is applied to the randomness to derive the session key $K = \mathsf{KDF}(r)$.
    \item \textbf{OAEP Padding:} Within the OAEP padding network, the message payload $m$ and the randomness $r$ are mutually masked via a Feistel-like structure using two hash functions (modeled as random oracles $G$ and $H$). Specifically, the network computes $s = m \oplus G(r)$ and $t = r \oplus H(s)$. These components are then concatenated to form a highly structured message block $X = s \parallel t$, whose length is matched to the RSA modulus.
    \item \textbf{RSA Encryption:} Finally, the encapsulator encrypts the formatted message block $X$ using the recipient's public RSA key $(N, e)$, yielding the final ciphertext $C = X^e \pmod N$.
\end{enumerate}

During decapsulation, after the receiver utilizes the RSA trapdoor (private key $d$) to decrypt the ciphertext and recover the padded payload $(s, t)$, it not only recovers the underlying message but, crucially, can deterministically extract the original masking randomness by computing $r = t \oplus H(s)$. Because this exact decoding mechanism guarantees flawless recovery of $r$, our anamorphic design can be directly integrated into legacy industrial systems deploying RSA-OAEP without requiring any modifications to the underlying cryptographic libraries.

\paragraph{Discrete Logarithm (DL) based Schemes.} 
Standard randomized DL-based schemes, such as textbook ElGamal is fundamentally non-randomness-recoverable. In these schemes, the ephemeral public key is computed as $R = g^r$ using a uniformly sampled randomness $r$. Extracting $r$ from $R$ requires solving the intractable Discrete Logarithm Problem (DLP).

However, DL-based KEMs inherently achieve the RR property when instantiated via modern FO transformations. A prominent standardized example is PSEC-KEM \cite{iso18033-2}, which is included in the ISO/IEC 18033-2 standard. In PSEC-KEM, the encapsulator derandomizes the process by deriving the ephemeral randomness deterministically from a hashed payload seed, i.e., $r = H(seed \parallel pk)$. During decapsulation, the receiver recovers the $seed$ using the DH shared secret and simply re-evaluates the hash to deterministically recover the exact randomness $r$. By bypassing the DLP bottleneck via FO-derandomization, schemes like PSEC-KEM serve as perfect $\mathsf{KEM}_{RR}$ substrates for our anamorphic design.

\paragraph{Lattice-based Schemes (NIST Post-Quantum Standard).} 
Following the NIST Post-Quantum Cryptography standardization, modern lattice-based KEMs, most notably ML-KEM (formerly CRYSTALS-Kyber, standardized as FIPS 203 \cite{fips203}), explicitly utilize post-2017 FO transformations with implicit rejection ($\mathsf{FO}^{\not\bot}$). In ML-KEM, the randomness vector $r$ used in the underlying LWE/MLWE encryption is deterministically derived via a random oracle query (SHAKE-256) taking the plaintext seed $m$ and the public key $pk$ as inputs: $r = G(m \parallel pk)$. During decapsulation, the decryptor recovers $m'$ and \textit{must} re-evaluate the hash to recover the exact randomness $r'$ for ciphertext verification via re-encryption. This mandatory re-encryption mechanism makes the FIPS 203 standard inherently randomness-recoverable, ensuring our framework is perfectly post-quantum ready without modifying the underlying encapsulation algorithms.

\subsection{Instantiations in the Standard Model}

While relying on random oracles is highly efficient in practice, constructing RR-KEMs strictly in the standard model is crucial for theoretical soundness. Fortunately, our framework can be robustly instantiated from well-established standard-model assumptions.

\paragraph{ATF and LTF-based Constructions.}
IND-CCA secure schemes constructed from Adaptive Trapdoor Functions (ATFs) \cite{KMO10} and Lossy Trapdoor Functions (LTFs) \cite{PW08} generally support explicit randomness recovery. In these constructions, the evaluation function $c = f_{pk}(x)$ takes an input $x$ that typically concatenates the message and the injected randomness (i.e., $x = m \parallel r$). The secret trapdoor allows the legitimate receiver to compute the strict inverse $f^{-1}_{sk}(c)$ and recover the complete pre-image $x$. Consequently, the decapsulation algorithm naturally extracts and outputs the exact randomness $r$ alongside the session key.

\paragraph{Lattices and LPN-based Constructions.}
In the post-quantum setting without random oracles, constructing CCA-secure schemes with extractable randomness is notoriously challenging. However, advancements in Lattices and Learning Parity with Noise (LPN) provide precise mathematical mechanisms for this. Notably, the standard-model CCA-secure schemes proposed in \cite{boyen2021secure} explicitly support randomness recovery. By leveraging specific trapdoor generation techniques or extractable hash proofs, these mechanisms can securely invert the LWE/LPN instances to recover the short error vectors, which serve as the encryption randomness $r$. By adopting these standard-model LPN/Lattice constructions, our anamorphic framework achieves robust post-quantum security entirely free from ROM heuristics.

\section{Conclusion}
In this paper, we formalized Public-Key and Symmetric-Key Anamorphic Key Encapsulation Mechanisms (PKAKEM and SKAKEM) and proposed generic constructions instantiable from any randomness-recoverable KEM. We rigorously proved their anamorphic security, ensuring that covert ciphertexts successfully bypass dictator detection. Crucially, our work achieves the realization of strong IND-CCA (sIND-CCA) security for covert messages within anamorphic systems. This guarantees strict confidentiality against chosen-ciphertext attacks, even if the dictator possesses the standard decapsulation key. With all security guaranties tightly proven in the standard model and demonstrating seamless compatibility with deployed standards such as RSA-OAEP and ML-KEM, our framework bridges the gap between provable security and real-world deployment.


\begin{thebibliography}{8}

\bibitem{rogaway2015moral}
  Rogaway, P.: The moral character of cryptographic work. Cryptology ePrint Archive (2015)

\bibitem{persiano2022anamorphic}
  Persiano, G., Phan, D.H., Yung, M.: Anamorphic encryption: Private communication against a dictator. In: Advances in Cryptology–EUROCRYPT 2022: 41st Annual International Conference on the Theory and Applications of Cryptographic Techniques, Proceedings, Part II. pp. 34–63. Springer (2022)

\bibitem{wang2023sender}
Wang, Y., Chen, R., Huang, X., Yang, G., Yung, M.: Sender-anamorphic encryption reformulated: Achieving robust and generic constructions. In: International Conference on the Theory and Application of Cryptology and Information Security. pp. 135-167. Springer (2023)

\bibitem{banfi2024anamorphic}
  Banfi, F., Gegier, K., Hirt, M., Maurer, U., Rito, G.: Anamorphic encryption, revisited. In: Annual International Conference on the Theory and Applications of Cryptographic Techniques. pp. 3-32. Springer (2024)

\bibitem{kutylowski2023self}
  Kutylowski, M., Persiano, G., Phan, D.H., Yung, M., Zawada, M.: The self-anti-censorship nature of encryption: On the prevalence of anamorphic cryptography. Proceedings on Privacy Enhancing Technologies \textbf{2023}(4), 170--183 (2023)

\bibitem{persiano2024public}
Persiano, G., Phan, D.H., Yung, M.: Public-key anamorphism in (CCA-secure) public-key encryption and beyond. In: Annual International Cryptology Conference. pp. 422--455. Springer (2024)

\bibitem{catalano2024anamorphic}
Catalano, D., Giunta, E., Migliaro, F.: Anamorphic encryption: New constructions and homomorphic realizations. In: Advances in Cryptology–EUROCRYPT 2024: 41st Annual International Conference on the Theory and Applications of Cryptographic Techniques, Proceedings, Part II. pp. 33--62. Springer (2024)

\bibitem{banerjee2025simple}
Banerjee, S., Pal, T., Rupp, A., Slamanig, D.: Simple Public Key Anamorphic Encryption and Signature using Multi-Message Extensions. Cryptology ePrint Archive (2025)

\bibitem{canetti2003relaxing}
Canetti, R., Krawczyk, H., Nielsen, J.B.: Relaxing chosen-ciphertext security. In: Advances in Cryptology–CRYPTO 2003: 23rd Annual International Cryptology Conference, Proceedings. pp. 565--582. Springer (2003)

\bibitem{faonio2020improving}
Faonio, A., Fiore, D.: Improving the efficiency of re-randomizable and replayable CCA secure public key encryption. In: International Conference on Applied Cryptography and Network Security. pp. 271--291. Springer (2020)

\bibitem{fujisaki1999secure}
Fujisaki, E., Okamoto, T.: Secure integration of asymmetric and symmetric encryption schemes. In: Annual International Cryptology Conference. pp. 537--554. Springer (1999)

\bibitem{canetti2004random}
Canetti, R., Goldreich, O., Halevi, S.: The random oracle methodology, revisited. Journal of the ACM (JACM) \textbf{51}(4), 557--594 (2004)

\bibitem{don2019security}
Don, J., Fehr, S., Majenz, C., Schaffner, C.: Security of the Fiat-Shamir transformation in the quantum random-oracle model. In: Annual International Cryptology Conference. pp. 356--383. Springer (2019)

\bibitem{canetti2004chosen}
Canetti, R., Halevi, S., Katz, J.: Chosen-ciphertext security from identity-based encryption. In: International Conference on the Theory and Applications of Cryptographic Techniques. pp. 207-222. Springer (2004)

\bibitem{choi2025unified}
Choi, W., Collins, D., Liu, X., Zikas, V.: A unified treatment of anamorphic encryption. Cryptology ePrint Archive (2025)

\bibitem{abe2005tag}
Abe, M., Gennaro, R., Kurosawa, K., Shoup, V.: Tag-KEM/DEM: A new framework for hybrid encryption and a new analysis of Kurosawa-Desmedt KEM. In: Annual International Conference on the Theory and Applications of Cryptographic Techniques. pp. 128--146. Springer (2005)

\bibitem{nagao2005universally}
Nagao, W., Manabe, Y., Okamoto, T.: A universally composable secure channel based on the KEM-DEM framework. In: Theory of Cryptography Conference. pp. 426--444. Springer (2005)

\bibitem{chen2020subvert}
Chen, R., Huang, X., Yung, M.: Subvert KEM to break DEM: practical algorithm-substitution attacks on public-key encryption. In: International Conference on the Theory and Application of Cryptology and Information Security. pp. 98--128. Springer (2020)

\bibitem{dent2003designer}
Dent, A.W.: A designer’s guide to KEMs. In: IMA International Conference on Cryptography and Coding. pp. 133--151. Springer (2003)

\bibitem{saito2018tightly}
Saito, T., Xagawa, K., Yamakawa, T.: Tightly-secure key-encapsulation mechanism in the quantum random oracle model. In: Annual International Conference on the Theory and Applications of Cryptographic Techniques. pp. 520--551. Springer (2018)

\bibitem{luby1988how}
Luby, M., Rackoff, C.: How to construct pseudorandom permutations from pseudorandom functions. SIAM Journal on Computing \textbf{17}(2), 373--386 (1988)

\bibitem{bellare2000security}
Bellare, M., Kilian, J., Rogaway, P.: The security of the cipher block chaining message authentication code. Journal of Computer and System Sciences \textbf{61}(3), 362--399 (2000)

\bibitem{dodis2012message}
Dodis, Y., Kiltz, E., Pietrzak, K., Wichs, D.: Message authentication, revisited. In: Annual International Conference on the Theory and Applications of Cryptographic Techniques. pp. 355--374. Springer (2012)

\bibitem{canetti2003forward}
Canetti, R., Halevi, S., Katz, J.: A forward-secure public-key encryption scheme. In: International Conference on the Theory and Applications of Cryptographic Techniques. pp. 255--271. Springer (2003)

\bibitem{moller2004public}
Möller, B.: A public-key encryption scheme with pseudo-random ciphertexts. In: European Symposium on Research in Computer Security. pp. 335--351. Springer (2004)

\bibitem{boneh2017constrained}
Boneh, D., Kim, S., Wu, D.J.: Constrained keys for invertible pseudorandom functions. Theory of Cryptography Conference. pp. 237--263. Springer (2017)

\bibitem{boyen2021secure}
Boyen, X., Izabachène, M., Li, Q.: Secure Hybrid Encryption in the Standard Model from Hard Learning. In: Post-Quantum Cryptography: 12th International Workshop, PQCrypto 2021, Daejeon, South Korea, July 20--22, 2021, Proceedings. pp. 399--418. Springer (2021)

\bibitem{rescorla2018transport}
Rescorla, E.: The Transport Layer Security (TLS) Protocol Version 1.3. RFC 8446 (2018)

\bibitem{blake2006transport}
Blake-Wilson, S., Nystrom, M., Hopwood, D., Mikkelsen, J., Wright, T.: Transport Layer Security (TLS) Extensions. RFC 4366 (2006)

\bibitem{eastlake2011transport}
Eastlake, D.: Transport Layer Security (TLS) Extensions: Extension Definitions. RFC 6066 (2011)

\bibitem{dierks2008transport}
Dierks, T., Rescorla, E.: The Transport Layer Security (TLS) Protocol Version 1.2. RFC 5246 (2008)

\bibitem{fujisaki2001rsa}
Fujisaki, E., Okamoto, T., Pointcheval, D., Stern, J.: RSA-OAEP is secure under the RSA assumption. In: Annual International Cryptology Conference. pp. 260--274. Springer (2001)

\bibitem{kiltz2017instantiability}
Kiltz, E., O'Neill, A., Smith, A.: Instantiability of RSA-OAEP under chosen-plaintext attack. Journal of Cryptology \textbf{30}(3), 889-919 (2017)

\bibitem{rfc3447}
Jonsson, J., Kaliski, B.: Public-Key Cryptography Standards (PKCS) \#1: RSA Cryptography Specifications Version 2.1. RFC 3447, RFC Editor (2003)

\bibitem{yao2005design}
Yao, F.F., Yin, Y.L.: Design and analysis of password-based key derivation functions. In: Cryptographers’ Track at the RSA Conference. pp. 245-261. Springer (2005)

\bibitem{krawczyk2010hmac}
Krawczyk, H., Eronen, P.: HMAC-based Extract-and-Expand Key Derivation Function (HKDF). RFC 5869 (2010)

\bibitem{krawczyk2010cryptographic}
Krawczyk, H.: Cryptographic extraction and key derivation: The HKDF scheme. In: Annual Cryptology Conference. pp. 631--648. Springer (2010)

\bibitem{iso18033-2}
ISO/IEC: ISO/IEC 18033-2:2006: Information technology -- Security techniques -- Encryption algorithms -- Part 2: Asymmetric ciphers. International Organization for Standardization, Geneva, Switzerland (2006)

\bibitem{fips203}
National Institute of Standards and Technology: Module-Lattice-Based Key-Encapsulation Mechanism Standard. Federal Information Processing Standards Publication (FIPS) 203, U.S. Department of Commerce (2024)

\bibitem{KMO10}
Kiltz, E., Mohassel, P., O'Neill, A.: Adaptive trapdoor functions and chosen-ciphertext security. In: Advances in Cryptology--EUROCRYPT 2010. pp. 673--692. Springer (2010)

\bibitem{PW08}
Peikert, C., Waters, B.: Lossy trapdoor functions and their applications. In: Proceedings of the 40th annual ACM symposium on Theory of computing (STOC). pp. 120--129. ACM (2008)





\end{thebibliography}
\end{document}